\documentclass[english]{article}
\usepackage{amsfonts}
\usepackage{amsmath}
\usepackage{amsthm} 
\usepackage{epsfig}

\usepackage[a4paper]{geometry}
\usepackage{setspace}
\newcommand{\e}{{\rm e}}
\newcommand{\E}{{\mathbb E}}
\newcommand{\Pa}{{\mathbb P}}

\newcommand{\R}{{\mathbb R}}

\newcommand{\Dcal}{{\mathcal D}}

\newcommand{\Pcal}{{\mathcal P}}

\newcommand{\Scal}{{\mathcal S}}
\newcommand{\Vcal}{{\mathcal V}}

\newtheorem{proposition}{Proposition}[section]
\newtheorem{lemma}[proposition]{Lemma}
\newtheorem{theorem}[proposition]{Theorem}
\newtheorem{definition}[proposition]{Definition}
\newtheorem{corollary}[proposition]{Corollary}
\newtheorem{remark}[proposition]{Remark}

\newtheorem{assumption}[proposition]{Assumption}

\newtheorem{exampleemph}[proposition]{Example}   
\newenvironment{example}{\begin{exampleemph}\begin{upshape}}{\end{upshape}\end{exampleemph}} 
\newtheorem{foo}[proposition]{Remarks}

\begin{document}

\title{Optimal Investment and Premium Policies under Risk Shifting and Solvency Regulation\footnote{We thank Ole von H\"afen, Julien Hugonnier, Christian Laux, Achim Wambach, and participants of the 2010 World Risk and Insurance Economics Congress (WRIEC) for helpful comments. Part of this research has been carried out within the project on "Dynamic Asset Pricing" of the National Centre of Competence in Research "Financial Valuation and Risk Management" (NCCR FINRISK). The NCCR FINRISK is a research instrument of the Swiss National Science Foundation.}}
\author{Damir Filipovi\'c\footnote{\'Ecole Polytechnique F\'ed\'erale de Lausanne and Swiss Finance Institute, Quartier UNIL-Dorigny, Extranef 218, CH - 1015 Lausanne, Switzerland; Email: damir.filipovic@epfl.ch}\and Robert
Kremslehner\footnote{Vienna University of Economics and Business, Department of Finance, Accounting and Statistics, Heiligenst\"adter Str. 46-48,
A - 1190 Wien, Austria; Email: robert.kremslehner@wu.ac.at}\and
Alexander Muermann\footnote{Vienna University of Economics and
Business, Department of Finance, Accounting and Statistics, Heiligenst\"adter
Str. 46-48, A - 1190 Wien, Austria; Email: alexander.muermann@wu.ac.at}}

\date{March 8, 2011}
\maketitle

\begin{abstract}
Limited liability creates a conflict of interests between policyholders and shareholders of insurance companies. It provides shareholders with incentives to increase the risk of the insurer's assets and liabilities which, in turn, might reduce the value policyholders attach to and premiums they are willing to pay for insurance coverage.

We characterize Pareto optimal investment and premium policies in this context and provide necessary and sufficient conditions for their existence and uniqueness. We then identify investment and premium policies under the risk shifting problem if shareholders cannot credibly commit to an investment strategy before policies are sold and premiums are paid. Last, we analyze the effect of solvency regulation, such as Solvency II or the Swiss Solvency Test, on the agency cost of the risk shifting problem and calibrate our model to a non-life insurer average portfolio.

\end{abstract}
\vspace{2ex}

\\[2ex]


\section{Introduction}

Risk shifting is a well-known agency problem in corporate finance
between shareholders and bondholders of a corporation (Jensen and
Meckling\ \cite{jm_76}, Green\ \cite{gre_84}, MacMinn\
\cite{macm_93}). The limited liability protection provides
incentives for management acting in the interests of shareholders to
select riskier projects at the expense of bondholders. If management
cannot credibly commit not to undertake those projects, bondholders
demand an appropriate interest rate differential which reflects the
agency cost.

Policyholders of a stock insurance company face a situation similar
to that of bondholders of a corporation. By paying premiums,
policyholders provide capital which is senior to equity but under
the investment decision of management acting in the interest of
shareholders. Limited liability provides incentives for management
to increase the risk of the insurer's assets and liabilities by, for
example, increasing the risk of the asset portfolio, selling
additional policies without a corresponding injection of equity, or
by changing the reinsurance arrangements or strategies in asset
liability management accordingly. This increase in the stock
insurer's risk may raise the insolvency probability of the
insurer and consequently reduce the value policyholders attach to
and the premium they are willing to pay for the insurance contract.

We provide a formal framework to study the conflict of interest
between risk-neutral shareholders of an insurance corporation and
its risk-averse policyholders. Shareholders have access to a
technology through which they can increase the risk of the insurer's
assets. We represent the technology by the possibility of investing
the total capital of the insurer, including equity and premium
payments, in a risky asset. Different risk profiles of the
insurer's assets are identified by different fractions of the
insurer's total capital invested in the risky asset. While
increasing the risk of the insurer's assets may serve shareholders'
interests, it may reduce the premium levels policyholders are
willing to pay and thereby the total capital available for
shareholders to invest.

In this context, we characterize the set of Pareto optimal
investment and premium policies. We show that, for any
policyholders' or shareholders' reservation utility level, there
exists a Pareto optimum and, under some mild assumption, it is unique. Moreover, we specify the necessary and
sufficient condition for the Pareto optimum.

We then investigate the risk shifting problem between shareholders
and policyholders, that is, the setting in which shareholders cannot
credibly commit to a specific investment strategy before
policies are sold and premiums are paid. This agency problem
generically leads to Pareto suboptimal investment policies and
corresponding premium levels. The investment technology implies that
the risk shifting problem admits only the boundary solution where the entire insurer's capital is invested
in the risky asset.

Last, we analyze the effect of solvency regulation in the context of
this agency problem. Solvency regulation imposes a constraint on the
set of possible investment and premium policies. We model the regulatory
constraint by some general convex risk measure and characterize the
corresponding solution. We show that there exists a unique
investment strategy and premium level that solves the risk shifting
problem under the regulatory constraint and analyze its effect on
the inefficiency of the risk shifting problem. Finally, we calibrate
our model to an European Economic Area non-life insurer average
portfolio taken from the QIS3 (Quantitative Impact Study 3)
Benchmarking Study~\cite{crof} of the Chief Risk Officer Forum under
the Solvency II standard model \cite{ceiops3,ceiops4} and illustrate
our analytical results.

Our paper relates to the literature on the risk shifting problem in
corporate finance (Jensen and Meckling\ \cite{jm_76}). Green\
\cite{gre_84} presents a formal model in which entrepreneurs decide
on the allocation of funds across two mutually non-exclusive projects where one project is
riskier than the other in the sense of Rothschild and Stiglitz\ \cite{rs_70}. The
investment technology exhibits some scale function which is strictly
concave in the amount invested in each project. The author shows
that the agency conflict leads to an overinvestment in the riskier
project and discusses the role of convertible bonds in eliminating
the risk shifting problem. MacMinn\ \cite{macm_93} analyzes the risk
shifting problem with two mutually exclusive projects but with a linear scale function. The author shows that a sufficiently high
level of leverage is necessary to induce the firm to switch from the
less risky to the riskier project and thereby to generate an agency
cost. Furthermore, there is a set of convertible contracts that
eliminate the risk shifting problem. We contribute to this
literature by analyzing the risk shifting problem in the insurance
setting where the two parties have differing costs of bearing risk
and which includes insurance losses as an additional source of risk.
Moreover, we examine the effect of solvency regulation on the agency
cost of the risk shifting problem. The investment technology in our
model can be interpreted as two projects which are mutually
non-exclusive (as in Green\ \cite{gre_84}) and exhibit a linear
scale function (as in MacMinn\ \cite{macm_93}).

Our paper also relates to the insurance literature that discusses
the conflicts of interest between shareholders and policyholders of
a stock insurer. Mayers and Smith\ \cite{ms_81},\ \cite{ms_88}
discuss an agency problem in the context of dividend policies. After
policies have been sold, shareholders have an incentive to increase
the value of their claim by raising dividends at the expense of
policyholders. They argue that the mutual organizational form of an
insurance company can help internalize this agency problem.
Doherty~\cite{doh_89} shows how an increase in the risk of
the insurer's asset portfolio increases the shareholders' position
at the expense of policyholders. We contribute to this literature by
providing a formal model to investigate this conflict of interest
and examine the effect of solvency regulation under the agency
problem. This framework allows us to determine the specific
conditions for existence and uniqueness of Pareto optimal risk
structures and premium levels, for the solutions under the agency
problem, and under solvency regulation.

Last, we contribute to the literature that analyzes the effect of
regulation in insurance markets. Munch and Smallwood\ \cite{mus_81}
and Finsinger and Pauly~\cite{fp_84} analyze the optimal amount of
shareholder capital and investment risk of an insurance company
under the assumption that shareholders cannot influence the risk
structure of the insurer's assets and that insurance premiums are
independent of the insurer's insolvency risk. Munch and
Smallwood~\cite{mus_81} argue that regulation may reduce the
insolvency risk if shareholder capital is low. Finsinger and
Pauly~\cite{fp_84}, however, argue that regulation of insurance
companies may be unnecessary in the long run with building up
reserves. McCabe and Witt~\cite{mw_80} and MacMinn and
Witt~\cite{mw_87} analyze the effect of different regulatory schemes
on the optimal investment and underwriting activity of an insurer
under the assumption that the demand function for insurance is
independent of the insurer's probability of insolvency. Rees et
al.~\cite{rgw_99} relax this assumption and show that if
policyholders are fully informed about the insurer's insolvency risk
then regulation serves no purpose. In our framework, shareholders
decide on the investment strategy, policyholders
are perfectly informed or, under the agency problem, have rational
expectations about the corresponding risk, and premiums
therefore depend on the implied insolvency risk. Moreover, we
examine the effect of solvency regulation on the agency cost.

The paper is structured as follows. We set up our model in
Section~\ref{secsetup} and characterize the set of Pareto optimal
policies in Section~\ref{secpareto}. In Section~\ref{secriskshift},
we examine the risk shifting problem and solvency regulation. In
Section~\ref{secnum}, we calibrate our model to data and illustrate
our results. We conclude in Section~\ref{secconcl}.

\section{Setup}\label{secsetup}

We consider a one-period economy with two agents, a policyholder and
a shareholder. The policyholder is endowed with some initial wealth
$w_0$ and faces a random loss $X$. His preferences are characterized
by some Bernoulli utility function $u:\R\to\R$.

The shareholder is risk-neutral. He owns a stock insurance company
with initial capital $c_0>0$ which offers full insurance coverage
for $X$ in exchange for a premium $p$. The shareholder has
access to a technology which allows him to increase the insurer's
risk. We represent this technology by an investment
opportunity in a risky asset that yields a random return $R$. The
shareholder thus decides on the fraction $\alpha\in [0,1]$ of the
total capital, $c_0+p$, to be invested in the risky asset. The
risk-free interest rate is assumed to be zero or, equivalently, all
values are in units of the risk-free numeraire. We assume that the shareholder has only access to the investment technology if he sells insurance. He may thus be willing to accept a negative premium $p>-c_0$ to gain access to the investment technology. We denote the set of investment and premium policies $(\alpha,p)$ by $\Pcal = [0,1]\times (-c_0,\infty)$.

The investment decision of the shareholder can be interpreted as an
allocation decision of funds across two projects, a risk-free
project with return zero and a risky project with random return $R$.
Moreover, the two projects are mutually
non-exclusive (as in Green\ \cite{gre_84}) and exhibit a linear
scale function (as in MacMinn\ \cite{macm_93}). For $\E\left[R \right]=0$, the project with
return $R$ is riskier in the sense of Rothschild and Stiglitz\
\cite{rs_70}. Allocating a higher fraction $\alpha$ of the total
capital, $c_0+p$, to the risky asset thus increases the risk of the
insurer's assets and liabilities. This can be achieved, for example,
by increasing the risk of the insurer's asset portfolio, by reducing
the duration matching of assets and liabilities, or by increasing
the attachment point of reinsurance contracts.

The end of the period surplus is given by $(c_0+p)(1+\alpha R)-X$.
If the surplus is negative, the insurance company is insolvent and
the shareholder is protected by limited liability. In this case, the
policyholder receives the remaining assets, $(c_0+p)(1+\alpha R)$.
Consequently, the terminal payoff to the shareholder equals
\[ ((c_0+p)(1+\alpha R)-X)^+ \]
while the terminal wealth of the policyholder is given by
\[ w_0-p-(X-(c_0+p)(1+\alpha R))^+ .\]
The corresponding utility of the shareholder and the policyholder as
a function of $\alpha$ and $p$ is
\[ U_{SH}(\alpha,p)=\E\left[ ((c_0+p)(1+\alpha R)-X)^+ \right] \]
and
\[ U_{PH}(\alpha,p)=\E\left[ u\left(w_0-p-(X-(c_0+p)(1+\alpha
R))^+\right)\right] ,\] respectively.
\begin{assumption}\label{ass1}
Throughout the paper, we make the following standing assumptions:
\begin{enumerate}
\item $u$ is increasing, concave, and twice differentiable on $\R$ with $u'>0$, $u''<0$, $\lim_{x\to
-\infty}u(x)=-\infty$ and $\lim_{x\to -\infty}u'(x)=\infty$.
  \item\label{ass1fxr}  $(X,R)$ takes values in $\R_+\times [-1,\infty)$ and admits a
jointly continuous density function $f(x,r)$.

\item The solvency event $\Scal(\alpha,p)=\{(c_0+p)(1+\alpha R)\ge X\}$ has
positive probability, $\Pa[\Scal(\alpha,p)]>0$, for all
$(\alpha,p)\in \Pcal$.

\item $u$ and $f$ are such that $U_{SH}$ and
$U_{PH}$ are real-valued and differentiable in some neighborhood of
${\Pcal}$, and the following formal manipulations (e.g.\
changing the order of differentiation and integration) are
meaningful.\footnote{For example, it is sufficient (but not
necessary) that $f$ has compact support.}
\end{enumerate}
\end{assumption}

Lemmas~\ref{lemvexcav}--\ref{lemUPHmax} in the appendix illustrate the qualitative behavior of the shareholder and policyholder utility function on $\Pcal$. In the sequel we will draw on these results without further mention. From Lemma~\ref{lemvexcav} we know that $U_{PH}(1,p)$ is strictly
concave in $p$. Hence there exists a unique critical premium level $p^{crit}_{PH}\in [-c_0,\infty)$ which maximizes the policyholder utility for $\alpha=1$,
\[ U_{PH}(1,p^{crit}_{PH})=\max_{p\in [-c_0,\infty)}\,U_{PH}(1,p).\]
We denote the
corresponding critical shareholder and policyholder utility levels by $\gamma^{crit}_{SH}=U_{SH}(1,p^{crit}_{PH})$ and $\gamma^{crit}_{PH}=U_{PH}(1,p^{crit}_{PH})$, and define the intervals\footnote{Note that $p^{crit}_{PH}=-c_0$ if and only if $\gamma^{crit}_{SH}=0$. }
\begin{align*}
  \Gamma_{SH}&=\begin{cases}
  (\gamma^{crit}_{SH},\infty)=(0,\infty),&\text{if $p^{crit}_{PH}=-c_0$,}\\ [\gamma^{crit}_{SH},\infty),&\text{if $p^{crit}_{PH}>-c_0$,}
\end{cases}\\
\Gamma_{PH}&=\begin{cases}
  (-\infty,\gamma^{crit}_{PH}),&\text{if $p^{crit}_{PH}=-c_0$,}\\
(-\infty,\gamma^{crit}_{PH}],&\text{if $p^{crit}_{PH}>-c_0$.}
\end{cases}
\end{align*}

\section{Pareto Optimal Investment and Premium Policies}\label{secpareto}

In this setup, we now examine optimal policies $(\alpha,p)\in\Pcal$ in the following sense:

\begin{definition}
The policy $(\alpha^\ast,p^\ast)\in \Pcal$ is Pareto
optimal if there does not exist any other policy $(\alpha,p)\in
\Pcal$ such that $U_{SH}(\alpha,p)\ge
U_{SH}(\alpha^\ast,p^\ast)$ and $U_{PH}(\alpha,p)\ge
U_{PH}(\alpha^\ast,p^\ast)$ with strict inequality for at least one
of them.
\end{definition}

We first show that Pareto optimality is equivalent to a constrained
optimization problem.

\begin{theorem}\label{thm1}
For any policy $(\alpha^\ast,p^\ast)\in \Pcal$, the following are
equivalent:
\begin{enumerate}
  \item\label{thm11} $(\alpha^\ast,p^\ast)$ is Pareto optimal.

\item\label{thm13} $(\alpha^\ast,p^\ast)$ solves the constrained optimization
  problem
  \begin{equation}\label{phopt}
  \begin{aligned}
    \max_{(\alpha,p)\in \Pcal}  & U_{PH}(\alpha,p)\\
    \text{s.t. }& U_{SH}(\alpha,p)\ge \gamma_{SH}
  \end{aligned}
  \end{equation}
  for the shareholder's reservation utility level $\gamma_{SH}=U_{SH}(\alpha^\ast,p^\ast)$, and $\gamma_{SH}\in\Gamma_{SH}$.

  \item\label{thm12} $(\alpha^\ast,p^\ast)$ solves the constrained optimization
  problem
  \begin{equation}\label{shopt}
  \begin{aligned}
    \max_{(\alpha,p)\in \Pcal}  & U_{SH}(\alpha,p)\\
    \text{s.t. }& U_{PH}(\alpha,p)\ge \gamma_{PH}
  \end{aligned}
  \end{equation}
  for the policyholder's reservation utility level $\gamma_{PH}=U_{PH}(\alpha^\ast,p^\ast)$, and $\gamma_{PH}\in\Gamma_{PH}$.

\end{enumerate}
Moreover, in either of the above optimization problems,
\eqref{phopt} and \eqref{shopt}, the respective reservation utility
constraint is binding.
\end{theorem}

\begin{proof}
{\ref{thm11}}$\Rightarrow${\ref{thm13}}: let $(\alpha^\ast,p^\ast)\in\Pcal$ be Pareto optimal. Then clearly $(\alpha^\ast,p^\ast)$ solves the constrained optimization problem \eqref{phopt}. It remains to be shown that $\gamma_{SH}\in\Gamma_{SH}$. If $(\alpha^\ast,p^\ast)=(1, p^{crit}_{PH})$ or if $p^{crit}_{PH}=-c_0$ then there is nothing to prove. So assume that $(\alpha^\ast,p^\ast)\neq (1, p^{crit}_{PH})\in\Pcal$. In view of Lemma~\ref{lemUPHmax} {\ref{lemUPHmax4}}, we have $U_{PH}(\alpha^\ast,p^\ast)< U_{PH}(1, p^{crit}_{PH})$. But then, by Pareto optimality of $(\alpha^\ast,p^\ast)$, we must have $\gamma_{SH}=U_{SH}(\alpha^\ast,p^\ast)>U_{SH}(1, p^{crit}_{PH})=\gamma^{crit}_{SH}$. This proves the claim.

{\ref{thm13}}$\Rightarrow${\ref{thm11}}: let $(\alpha^\ast,p^\ast)\in\Pcal$ be a maximizer of \eqref{phopt}. We argue by contradiction and assume that $(\alpha^\ast,p^\ast)$ is not Pareto optimal. Then
there exists some policy $(\bar{\alpha},\bar{p})\in\Pcal$ such that
$U_{PH}(\bar{\alpha},\bar{p})\ge U_{PH}(\alpha^\ast,p^\ast)$ and
$U_{SH}(\bar{\alpha},\bar{p})\ge \gamma_{SH}=U_{SH}(\alpha^\ast,p^\ast)$ with strict inequality
for at least one of them. If $U_{PH}(\bar{\alpha},\bar{p})>
U_{PH}(\alpha^\ast,p^\ast)$ then clearly $(\alpha^\ast,p^\ast)$ cannot be an optimizer of \eqref{phopt}. Hence we can assume that $U_{SH}(\bar{\alpha},\bar{p})>\gamma_{SH}$ and $U_{PH}(\bar{\alpha},\bar{p})= U_{PH}(\alpha^\ast,p^\ast)$. Then there exists a neighborhood $O$ of
$(\bar{\alpha},\bar{p})$ in $\Pcal$ such that $U_{SH}(\alpha,p)\ge\gamma_{SH}$ for all $(\alpha,p)\in O$. By Lemma~\ref{lemgrad} below, $\nabla U_{PH}(\bar{\alpha},\bar{p})\neq 0$, and $\partial_p U_{PH}(\bar{\alpha},\bar{p})<0$ if $\bar{\alpha}=0$ in particular. Moreover, if $\bar{\alpha}=1$ then $U_{SH}(1,\bar{p})>\gamma_{SH}\ge\gamma^{crit}_{SH}=U_{SH}(1, p^{crit}_{PH})$ implies $\bar{p}> p^{crit}_{PH}$ and thus $\partial_p U_{PH}(\bar{\alpha},\bar{p})<0$ again. Hence in either case we conclude that there exists some $(\alpha,p)\in O$ with
$U_{PH}(\alpha,p)>U_{PH}(\bar{\alpha},\bar{p})\ge
U_{PH}(\alpha^\ast,p^\ast)$, whence $(\alpha^\ast,p^\ast)$
does not solve \eqref{phopt}, which is absurd. Hence $(\alpha^\ast,p^\ast)$ is Pareto optimal.
 Moreover, this shows that the
reservation utility constraint $U_{SH}(\alpha,p)\ge \gamma_{SH}$ is binding.

The equivalence {\ref{thm11}}$\Leftrightarrow${\ref{thm12}} follows
similarly but simpler, since $\partial_p U_{SH}(\alpha,p)>0$ by
Lemma~\ref{lemgrad}.

\end{proof}

Pareto optimal investment and premium policies can thus be generated
 by a take-it-or-leave-it offer either of the policyholder to the shareholder (optimization problem
\eqref{phopt}) or of the shareholder to the
policyholder (optimization problem \eqref{shopt}). The participation constraint is binding in either
case because both the policyholder's and shareholder's preferences
are locally non-satiated.

Note that part {\ref{thm13}} of Theorem~\ref{thm1} implies that there exists no Pareto optimal policy for a shareholder utility level $\gamma_{SH}\notin \Gamma_{SH}$. In fact, as for the existence of Pareto optimal policies, we have the following theorem.

\begin{theorem}\label{thm2}
For any reservation utility level $\gamma_{SH}\in \Gamma_{SH}$ and $\gamma_{PH}\in\Gamma_{PH}$, respectively, there exists at least one
Pareto optimum $(\alpha^\ast,p^\ast)\in \Pcal$ with $U_{SH}(\alpha^\ast,p^\ast)=\gamma_{SH}$ and $U_{PH}(\alpha^\ast,p^\ast)=\gamma_{PH}$, respectively. It
satisfies the first order condition
\begin{equation}\label{foceq}
 \E\left[ R\, u'\left(w_0-p^\ast-(X-(c_0+p^\ast)(1+\alpha^\ast R))^+\right)\right]\begin{cases}
  \le 0, &\text{if $\alpha^\ast=0$,}\\
  = 0, &\text{if $0<\alpha^\ast<1$,}\\
  \ge 0, &\text{if $\alpha^\ast=1$}.
\end{cases}
\end{equation}
Moreover, for any $\alpha^\ast\in [0,1]$ there exists at most one
Pareto optimum.
\end{theorem}

\begin{proof}

In view of Theorem~\ref{thm1}, it is enough to consider the
optimization problems~\eqref{phopt} and \eqref{shopt}, respectively. First note that, in view of Lemmas~\ref{lemvexcav} and \ref{lemgrad}, $U_{PH}(\alpha,p)$ is strictly concave and $U_{SH}(\alpha,p)$ is strictly increasing in $p$, for all $\alpha\in [0,1]$. Hence, for
any fixed $\alpha\in [0,1]$, there can be at most one Pareto
optimum. Moreover, by Lemma~\ref{lemUPHmax}, we have $\Gamma_{SH}\subseteq U_{SH}(\Pcal)$ and $\Gamma_{PH}=U_{PH}(\Pcal)$. Hence the constraint sets in \eqref{phopt} and \eqref{shopt} are non-empty. Lemma~\ref{lemUPHmax} {\ref{lemUPHmax3}} implies that for any $\gamma_{PH}\in\Gamma_{PH}$ the level set $\{ U_{PH}\ge \gamma_{PH}\}\subset\Pcal$ is compact in $\Pcal$. Since $U_{PH}$ and $U_{SH}$ are continuous on $\Pcal$, we conclude that the maximum in both optimization problems \eqref{phopt} and \eqref{shopt}, and thus the Pareto optimum at the respective reservation utility level, is attained in $\Pcal$.

For the derivation of the first order condition, it is convenient to
introduce the following diffeomorphism:
\begin{align*}
 \Pcal&\to   \Vcal=\{(v,w)\mid 0< v\le w<\infty\}\\
 (\alpha,p)&\mapsto (v,w)=((c_0+p)\alpha,c_0+p).
\end{align*}
Note that $w$ is the total asset value of the insurer and $v$ is the
money invested in the stock market. The corresponding utility of the
shareholder and the policyholder as a function of the new
coordinates $(v,w)$ is
\begin{equation}\label{vwnotation}
\begin{aligned}
  V_{SH}(v,w)&=\E\left[ (w+v R-X)^+ \right]\\
V_{PH}(v,w)&=\E\left[ u\left(w_0+c_0-w-(X-w-v R)^+\right)\right],
\end{aligned}
\end{equation}
so that $V_{SH}(v,w)=U_{SH}(\alpha,p)$ and
$V_{PH}(v,w)=U_{PH}(\alpha,p)$. For simplicity of notation, we use
the same letter $\Scal(v,w)=\Scal(\alpha ,p )$ for the respective
solvency event.

We note that $V_{SH}$ is a convex and $V_{PH}$ is a
concave function jointly in $(v,w)$. In contrast, $U_{SH}$ and
$U_{PH}$ do not share these properties as functions jointly in
$(\alpha,p)$ in general.

By Lemma~\ref{lemgradV}, we have $\partial_w V_{PH}<0$. Hence, for any $\gamma_{SH}\in \Gamma_{SH}$, the implicit function theorem yields a
continuously differentiable function $W:I\to (0,\infty)$ on some
interval $I\subset\R_+$ with $(v,W(v))\in\Vcal$,
$V_{SH}(v,W(v))=\gamma_{SH}$, and
\begin{equation}\label{Wpeq}
  W'(v)=-\frac{\partial_v V_{SH}(v,W(v)) }{\partial_w V_{SH}(v,W(v))}=- \frac{\E\left[R\,1_{\Scal(v,w)}\right]}{\Pa \left[{\Scal(v,w)}\right]}.
\end{equation}
Since for every fixed $\alpha\in [0,1]$ the function $V_{SH}(\alpha w,w)=U_{SH}(\alpha,w-c_0)$ is strictly increasing in $w$ and maps the interval $(0,\infty)$ onto itself, we can assume that $I=[0,v']$ for some $v'>0$, and that $W(0)=0$ and $W(v')=v'$.

We now characterize the critical points for the policyholder utility
function along the level curve $(v,W(v))$. A calculation shows
\begin{equation}\label{foceqpre}
 \begin{aligned}
  \frac{d}{dv} V_{PH}(v,W(v)) &=\partial_v V_{PH}(v  ,W(v)  )+W'(v)\partial_w
V_{PH}(v,W(v))\\
&=\E\left[R\,u'(w_0+c_0-X+vR)\,1_{{\Scal(v,w)}^c}\right] \\
&\quad+\frac{\E\left[R\,1_{\Scal(v,w)}  \right]}{\Pa \left[{\Scal(v,w)} \right]} u'(w_0+c_0-w)\Pa \left[{\Scal(v,w)} \right]\\
&= \E\left[R\,u'\left(w_0+c_0-W(v)-(X-W(v)-v R)^+\right) \right].
\end{aligned}
\end{equation}
Hence any Pareto optimal $(v^\ast,w^\ast)\in\Vcal$ satisfies
\[\E\left[R\,u'\left(w_0+c_0-w^\ast-(X-w^\ast-v^\ast R)^+\right)\right]
\begin{cases}
\le 0,&\text{if $v^\ast=0$,}\\=0&\text{if $0<v^\ast<w^\ast$,}\\
\ge 0,&\text{if $v^\ast=w^\ast$.}
\end{cases}\]
This proves \eqref{foceq}.
\end{proof}

In this theorem, we have shown the existence of Pareto optimal
investment and premium policies for any admissible reservation
utility level of the policyholder or the shareholder. This result is
thus valid for different degrees of competition in the insurance
market which can be represented by different reservation utility
levels. A higher degree of competition in the insurance market is
reflected by a lower reservation utility level of the shareholder.
In a perfectly competitive market, the shareholder reservation
utility level is given by $\gamma_{SH}=c_0$, derived from his
outside option of not selling insurance.

In the following theorem, we specify the condition under which the
Pareto optimum is unique and under which the first order
condition~\eqref{foceq} is also sufficient.

\begin{theorem}\label{thm3}
Assume that $\E[R]\neq 0$ or that, for any $(\alpha,p)\in \Pcal$, either the insolvency event has positive probability, $\Pa\left[{\Scal(\alpha,p)}^c\right]>0$, or $\E\left[R\,1_{\Scal(\alpha,p)}\right]\neq 0$. Then for any reservation utility level $\gamma_{SH}\in \Gamma_{SH}$ and $\gamma_{PH}\in \Gamma_{PH}$, respectively, there exists a unique Pareto
optimum $(\alpha^\ast,p^\ast)\in \Pcal$. Moreover, the
first order condition~\eqref{foceq} is also sufficient for Pareto
optimality of $(\alpha^\ast,p^\ast)\in \Pcal$.
\end{theorem}

\begin{proof}
We use the $(v,w)$-coordinates as introduced in
the proof of Theorem~\ref{thm2}. Fix $\gamma_{SH}\in \Gamma_{SH}$, and let $W:I\to (0,\infty)$ be the corresponding level curve as in the
proof of Theorem~\ref{thm2}. In view of \eqref{Wpeq} and \eqref{foceqpre}, the second derivative of $V_{PH}(v,W(v))$ equals\footnote{The differentiation under the expectation
sign is justified by Assumption~\ref{ass1}, see the proof of
Lemma~\ref{lemgrad} below.}
\begin{align*}
 \frac{d^2}{dv^2} V_{PH}(v,W(v))  &=
 -W'(v ) \E\left[R\,1_{{\Scal(v ,W(v ))}}
\right]u''\left(w_0+c_0-W(v )\right)\\
&\quad + \E\left[R^2\,u''\left(w_0+c_0-X+v  R
\right)\,1_{{\Scal(v ,W(v ))}^c}
\right]\\
&= \frac{\E\left[R\,1_{\Scal(v,W(v))}\right]^2}{\Pa \left[{\Scal(v,W(v))}\right]}u''\left(w_0+c_0-W(v )\right)\\
 &\quad+ \E\left[R^2\,u''\left(w_0+c_0-X+v  R
\right)\,1_{{\Scal(v ,W(v ))}^c}
\right],
\end{align*}
which under the assumption of the theorem is negative. Indeed, both summands on the right hand side are non-positive. If the insolvency event has positive probability, $\Pa[{\Scal(v ,W(v ))}^c]>0$,  then the second term is negative. Otherwise $\Pa[{\Scal(v ,W(v ))}]=1$ and thus $\E\left[R\,1_{\Scal(v,W(v))}\right]^2=\E\left[R\right]^2>0$, so that the first term is negative.
Hence $V_{SH}(v,W(v))$ is strictly concave in $v\in I$. This proves the theorem.
\end{proof}

In the proof of this theorem, it is shown that, under very mild assumptions, the policyholder's utility as a function of
$\alpha$ along a shareholder's level curve, $U_{SH}(\alpha,p)=
\gamma_{SH}$, can only assume three shapes. Either it is strictly
decreasing in which case $\alpha^\ast=0$ is the Pareto optimal
investment policy. Or it is strictly increasing in which case
$\alpha^\ast=1$ is the Pareto optimal investment policy. Last, it
can take a unique inner maximum $0<\alpha^\ast<1$ and is strictly
increasing to the left and strictly decreasing to the right. In each
of these three cases, the Pareto optimal premium policy $p^\ast$ is
uniquely defined by $U_{SH}(\alpha^\ast,p^\ast)= \gamma_{SH}$.

\begin{remark}\label{rem3}
  Along similar arguments as in the proof of Theorem~\ref{thm3}, but under the more stringent assumption that $\E\left[ R\,1_{\Scal(\alpha,p)}\right]>0$ for all
  $(\alpha,p)\in \Pcal$, one can show that also the shareholder's utility along a policyholder's level curve, $U_{PH}=
\gamma_{PH}$, can only assume the afore mentioned three shapes: increasing, decreasing, or first increasing and then decreasing.\footnote{Since $\partial_w V_{PH}<0$, the implicit function theorem yields a $C^1$-policyholder utility level curve $v\mapsto W(v)$ in $\Vcal$. Some calculations show that $\frac{d}{dv} V_{SH}(v,W(v)) = \frac{\E\left[R\,u'\left(w_0+c_0-W(v)-(X-W(v)-v R)^+\right)
\right]}{u'(w_0+c_0-W(v))}$. Further one can show that for critical points $v^\ast$ where $\frac{d}{dv} V_{SH}(v^\ast,W(v^\ast))=0$ we have $W'(v^\ast)<0$ and
\begin{align*}
\frac{d^2}{dv^2} V_{SH}(v,W(v))|_{v=v^\ast} &=
 -W'(v^\ast )\frac{u''\left(w_0+c_0-W(v^\ast )\right)\E\left[R\,1_{{\Scal(v^\ast ,W(v^\ast ))}}
\right]}{u'(w_0+c_0-W(v^\ast ))}\\
&\quad +\frac{\E\left[R^2\,u''\left(w_0+c_0-X+v^\ast  R
\right)\,1_{{\Scal(v^\ast ,W(v^\ast ))}^c}
\right]}{u'(w_0+c_0-W(v^\ast ))}.
\end{align*}
 From this we conclude that $\frac{d^2}{dv^2} V_{SH}(v,W(v))|_{v=v^\ast}<0$. Hence every critical point is a local maximum. This shows that $V_{SH}(v,W(v))$ has the desired properties.} Again, the respective Pareto optimal policy is
uniquely defined.
\end{remark}

We terminate this section with a generic example where there exists no inner Pareto optimum $(\alpha^\ast,p^\ast)\in\Pcal$ with investment policy $\alpha^\ast>0$.

\begin{example}
 For this example we assume that
\begin{enumerate}
  \item $X$ and $R$ are independent,
  \item $\E[R]=0$,
  \item $\Pa\left[\Scal(\alpha,p)^c \right]>0$ for all $(\alpha,p)\in \Pcal$.
\end{enumerate}
In particular, the assumption in Theorem~\ref{thm3} is satisfied. Now fix an arbitrary policy $(\alpha,p)\in \Pcal$ with $\alpha>0$. We claim that
\begin{equation}\label{hilfsin}
   \E\left[ R\, u'\left(w_0-p-(X-(c_0+p)(1+\alpha  R))^+\right)\right]<0.
\end{equation}
This together with the first order condition~\eqref{foceq} and Theorem~\ref{thm3} then implies that the set of all Pareto optimal policies in $\Pcal$ is given by $\{(0,p^\ast)\mid p^\ast\in (-c_0,\infty)\}$. For the proof of \eqref{hilfsin} we first observe that
\[ u'\left(w_0-p-(X-(c_0+p)(1+\alpha  R))^+\right)\begin{cases}
  \le u'\left(w_0-p -(X-(c_0+p))^+\right)&\text{on $\{R\ge 0\}$}\\
  \ge u'\left(w_0-p -(X-(c_0+p))^+\right)&\text{on $\{R<0\}$}
\end{cases}\]
with strict inequalities on $\Scal(\alpha,p)^c\cap\{R>0\}$ and $\Scal(\alpha,p)^c\cap\{R<0\}$, respectively. Hence
\[ R\, u'\left(w_0-p-(X-(c_0+p)(1+\alpha  R))^+\right)\le  R\, u'\left(w_0-p-(X-(c_0+p))^+\right)\]
with strict inequality on $\Scal(\alpha,p)^c\cap\{R\neq 0\}$. Since $\Pa\left[\Scal(\alpha,p)^c\cap\{R\neq 0\}\right]=\Pa\left[\Scal(\alpha,p)^c \right]>0$, by the above assumptions, we obtain
\begin{align*}
  \E\left[ R\, u'\left(w_0-p-(X-(c_0+p)(1+\alpha  R))^+\right)\right]&< \E\left[ R\, u'\left(w_0-p-(X-(c_0+p) )^+\right)\right]\\
  &=\E\left[ R\right] \E\left[u'\left(w_0-p-(X-(c_0+p) )^+\right)\right]\\
  &=0,
\end{align*}
which proves \eqref{hilfsin}.
\end{example}

\section{Risk Shifting and Solvency Regulation}\label{secriskshift}

In this section, we focus in our analysis of the risk shifting
problem on a competitive insurance market. The shareholder is
therefore held at his reservation utility level which is derived
from his outside option of not selling insurance. We assume that the
shareholder has only access to the investment technology if he sells
insurance. The shareholder's participation constraint is thus given
by $U_{SH}(\alpha,p)\ge c_0$.\footnote{Although we focus on a
competitive insurance market, we show that all results also hold
under the dual problem with the policyholder's participation
constraint. The results are thus valid for different degrees of
competition.}
\begin{assumption}\label{assrs}
Throughout this section, we make the following standing assumptions:
\begin{enumerate}
  \item\label{assrs1} $c_0\ge \gamma^{crit}_{SH}$,

  \item\label{assrs2} $\E[R\mid X\le x]>0$ for all $x\in(0,\infty)$.
\end{enumerate}
\end{assumption}

\begin{lemma}\label{lemUPalpha}
Assumption \ref{assrs}{\ref{assrs2}} is equivalent to $\E\left[R\,1_{\Scal(\alpha,p)}\right]>0$ for all $(\alpha,p)\in \Pcal$, which again is equivalent to
\begin{equation}\label{thm2assnewfullX}
  \partial_\alpha U_{SH}(\alpha,p)>0\quad\text{for all
  $(\alpha,p)\in \Pcal$.}
\end{equation}
In this case, the assumption in Theorem~\ref{thm3} is satisfied.
\end{lemma}

\begin{proof}
It follows by inspection that $\E\left[R\,1_{\{X\le c_0+p\}}\right]\le\E\left[R\,1_{\Scal(\alpha,p)}\right]$ for all $(\alpha,p)\in \Pcal$. In view of Lemma~\ref{lemgrad}, Assumption \ref{assrs}{\ref{assrs2}} implies
\eqref{thm2assnewfullX}. Conversely, that \eqref{thm2assnewfullX} implies {\ref{assrs2}}
follows from setting $\alpha=0$ in \eqref{thm2assnewfullX}. Moreover, it follows that the assumption in Theorem~\ref{thm3} is satisfied.
\end{proof}

\begin{remark}\label{remexpreturn}
Assumption \ref{assrs}{\ref{assrs2}} states that the conditional expected
return of the risky asset given bounded insurance losses is positive.
From a risk management point of view, it is important to note that
this moderate assumption is compatible with stress scenarios of
negative expected returns under catastrophic insurance losses. A
sufficient condition for {\ref{assrs2}} is
\[ \E\left[ R\mid X=x\right]>0\quad \text{for all $ {\rm ess \,inf}\, X<x<{\rm ess \,sup}\, X
  $}.\]
In particular, this sufficient condition holds if the risky asset
has a positive unconditional expected return, $\E[R]>0$, and $R$ and
$X$ are independent.
\end{remark}

We now consider the following sequence of events. The policyholder
pays a premium $p$ in exchange for full insurance coverage of $X$.
The shareholder then decides on the fraction $\alpha$ of the total
capital, $c_0+p$, to be invested in the risky asset.

If the shareholder can commit to an investment policy $\alpha$
before the insurance premium $p$ is paid, then the optimal
investment and premium policy is given by the solution to the
following optimization problem:

\begin{equation}\label{phoptcomp}
  \begin{aligned}
    \max_{(\alpha,p)\in \Pcal}  & U_{PH}(\alpha,p)\\
    \text{s.t. }& U_{SH}(\alpha,p)\ge \gamma_{SH}
  \end{aligned}
  \end{equation}
  for the shareholder's reservation utility level $\gamma_{SH}=c_0$. By Theorem~\ref{thm1} and Assumption~\ref{assrs},
the solution to \eqref{phoptcomp} is Pareto optimal.

We now examine the situation in which the shareholder cannot
credibly commit to an investment policy $\alpha$ before the
insurance premium $p$ is paid. That is, after the policyholder has
paid the insurance premium $p$, the shareholder chooses $\alpha$ by
solving the following optimization problem:

\begin{equation}\label{shoptars}
  \begin{aligned}
    \max_{\alpha\in [0,1]}  & U_{SH}(\alpha,p).\\
  \end{aligned}
  \end{equation}
By Lemma~\ref{lemUPalpha}, for any fixed $p\in (-c_0,\infty)$,
$U_{SH}(\alpha,p)$ is strictly increasing in $\alpha$. Hence only the boundary
solution is possible, i.e.\ ${\rm arg\, max}_{\alpha}\,
U_{SH}(\alpha,p)=1$. The shareholder invests the
entire total capital in the risky asset.\footnote{There is
only the boundary solution under the risk shifting problem in our
setting due to two features of the investment technology. First, any
fraction $\alpha\in[0,1]$ can be invested in the risky asset. Second, the investment technology
exhibits a linear scale function, i.e., there are no costs other
than the agency cost involved in shifting risk. A strictly concave
scale function, as assumed in Green\ \cite{gre_84} with mutually
non-exclusive projects, would imply that the shareholder overinvests
in risk but not necessarily up to the boundary.}

The policyholder has rational expectations about this investment
strategy and the premium $p$ has to satisfy the incentive
compatibility constraint that the shareholder sets the investment
strategy according to optimization problem~\eqref{shoptars}. We can
thus formalize the agency problem by the following constrained
optimization problem:

\begin{equation}\label{shoptrs}
  \begin{aligned}
    \max_{(\alpha,p)\in \Pcal}  & U_{PH}(\alpha,p)\\
    \text{s.t. }& U_{SH}(\alpha,p)\ge \gamma_{SH},\\
    &\alpha\in {\rm arg\, max}_{\alpha'}\, U_{SH}(\alpha',p)
  \end{aligned}
  \end{equation}
for the shareholder's reservation utility level $\gamma_{SH}=c_0$.
In view of Theorem~\ref{thm1} {\ref{thm13}}, any solution of the risk shifting problem
\eqref{shoptrs} is generically Pareto suboptimal.

In the following theorem, we characterize the solution and show that
it can be equivalently derived from maximal shareholder utility constraining on the
policyholder's reservation utility level.\footnote{The results thus
also apply to other degrees of competition, including a monopolistic
insurance market which is given by the policyholder's reservation
utility level derived from his outside option of not buying
insurance, i.e.\ $\gamma_{PH}=\E\left[ u(w_0-X)\right]$. Note that this is covered by the theorem since $\E\left[ u(w_0-X)\right]<\E\left[ u(w_0+c_0-X)\right]\le \gamma^{crit}_{PH}$.}

\begin{theorem}\label{thmrs}
For any reservation utility level $\gamma_{SH}\in \Gamma_{SH}$ and $\gamma_{PH}\in\Gamma_{PH}$, respectively, there
exists a unique solution $(\bar{\alpha},\bar{p})$ in
$\Pcal$ to the constrained optimization problem \eqref{shoptrs} and to the dual problem
\begin{equation}\label{phoptrs}
  \begin{aligned}
    \max_{(\alpha,p)\in \Pcal}  & U_{SH}(\alpha,p)\\
    \text{s.t. }& U_{PH}(\alpha,p)\ge \gamma_{PH},\\
    &\alpha\in {\rm arg\, max}_{\alpha'}\, U_{SH}(\alpha',p)
  \end{aligned}
  \end{equation}
respectively. This solution satisfies $\bar{\alpha}=1$, and $U_{SH}(1,\bar{p})= \gamma_{SH}$ and $U_{PH}(1,\bar{p})= \gamma_{PH}$, respectively. It is Pareto optimal if and only if $\E\left[ R\, u'\left(w_0-\bar{p}-(X-(c_0+\bar{p})(1+
R))^+\right)\right]\ge  0$.

Moreover, any solution $(\bar{\alpha},\bar{p})$ to \eqref{shoptrs} with $\gamma_{SH}=U_{SH}(\bar{\alpha},\bar{p})\in \Gamma_{SH}$ is a solution to \eqref{phoptrs} with $\gamma_{PH}=U_{PH}(\bar{\alpha},\bar{p})\in  \Gamma_{PH}$, and vice versa.

\end{theorem}

\begin{proof}
As seen below \eqref{shoptars}, for any fixed $p\in(-c_0,\infty)$, we have ${\rm
arg\, max}_{\alpha'}\, U_{SH}(\alpha',p)=1$. Hence problem \eqref{shoptrs} boils down to optimizing the continuous function $U_{PH}(1,p)$ on the interval $[\bar{p},\infty)$ with $\bar{p}$ determined by $U_{SH}(1,\bar{p})=\gamma_{SH}$. Note that $\gamma_{SH}\in\Gamma_{SH}$ is equivalent to $\bar{p}\ge p^{crit}_{PH}$. Thus, $U_{PH}(1,p)$ is strictly decreasing in $p\in[\bar{p},\infty)$. The optimizer of \eqref{shoptrs} is thus $(1,\bar{p})$. The first order condition for Pareto optimality follows from Theorem~\ref{thm3}, see Lemma~\ref{lemUPalpha}. A similar argument applies to problem \eqref{phoptrs}. We conclude, in particular, that any solution $(1,\bar{p})$ to either \eqref{shoptrs} with $\gamma_{SH}\in \Gamma_{SH}$ or \eqref{phoptrs} with $\gamma_{PH}\in \Gamma_{PH}$ satisfies $\bar{p}\ge p^{crit}_{PH}$. The last statement thus follows by the fact that $U_{PH}(1,p)$ and $U_{SH}(1,p)$ are strictly decreasing and increasing in $p\in [p^{crit}_{PH},\infty)$, respectively.
\end{proof}

The inefficiency of the investment and premium policy arises if
shareholders cannot credibly commit to an investment strategy before
premiums are paid. Any credible commitment device for shareholders
would thus increase welfare. In the context of corporate finance,
Green\ \cite{gre_84} and MacMinn\ \cite{macm_93} have shown that
issuing convertible bonds can provide such a commitment device and
eliminate the risk shifting problem.

In our context of insurance there are also various contractual and
organizational features that might provide some form of commitment
device. Shareholders can limit the insurer's risk exposure by
including restrictions on their investment and dividend policies in
their corporate charter. Shareholders could also transfer parts of
their assets to an escrow account. These restrictions limit the
extent to which shareholders can increase the insurer's risk
structure and thereby reduce the inefficiency caused by the risk
shifting problem. Along the lines of convertible bonds,
participating policies reduce the benefit to shareholders from
increasing the insurer's risk. The implied benefit of
partially overcoming the risk shifting problem might outweigh the
cost of higher exposure for policyholders to the insurer's risk.
Last, changing the organizational form to a mutual form would
eliminate the risk shifting problem. Under the mutual form, owners
who decide on the insurer's risk structure coincide with providers
of capital. The incentive problem is thereby eliminated at the cost
of less diversified owners.

In the following, we take the insurance contract as given and explore the effect of solvency regulation on the agency cost of the risk shifting problem. The regulator assesses
the riskiness of the annual loss
\[ L(\alpha,p)=c_0 - \left((c_0+p)(1+\alpha R)-X\right)=-(c_0+p)\alpha R +
X-p\] by means of a risk measure $\rho$. The respective regulatory
requirement is that the available capital, $c_0$, be greater than
the required capital, $\rho\left(L(\alpha,p)\right)$:
\[ \rho\left(L(\alpha,p)\right)\le c_0.\]
This capital requirement restricts the set of feasible investment
and premium policies. Solvency capital requirement can thus be
interpreted as a commitment device for shareholders imposed by the
regulator.

The risk shifting problem \eqref{shoptrs} under this additional regulatory
constraint is as follows
\begin{equation}\label{shoptrsreg}
  \begin{aligned}
    \max_{(\alpha,p)\in \Pcal}  & U_{PH}(\alpha,p)\\
    \text{s.t. }& U_{SH}(\alpha,p)\ge \gamma_{SH},\\
    &\alpha\in {\rm arg\, max}_{\alpha'}\, U_{SH}(\alpha',p)\\
    &\qquad\qquad \text{s.t. } \rho(L(\alpha',p))\le c_0
  \end{aligned}
  \end{equation}
for the shareholder's reservation utility level $\gamma_{SH}=c_0$.

We make the following standard assumptions for risk measures, see
e.g.\ McNeil et al.\ \cite{mfe}.

\begin{assumption}\label{assrho}
Throughout, $\rho$ satisfies the following conditions:
\begin{enumerate}
  \item $\rho$ is cash-invariant, that is, $\rho(L+c)=\rho(L)+c$ for
  any constant cash amount $c\in \R$ and random loss $L$.

  \item $\rho$ is convex, that is,
\[ \rho(\lambda L+(1-\lambda)L')\le
\lambda\rho(L)+(1-\lambda)\rho(L') \] for all $\lambda\in [0,1]$ and
random losses $L,\,L'$.

\item\label{assrhomon} $\rho$ is monotone, that is, $\rho(L)\le\rho(L')$ if $L\le L'$.

  \item\label{assrho2} $\rho(L(\alpha,p))$ is continuous in $\alpha\in [0,1]$ for
  all $p\in (-c_0,\infty)$.
\end{enumerate}
\end{assumption}

The cash-invariance property of $\rho$ is motivated by its
interpretation as regulatory capital requirement. Adding a
deterministic cash amount $c$ to the position, the capital
requirement is reduced by the same amount. The economic idea behind
the convexity assumption of $\rho$ is that diversification by means
of combining risks reduces overall risk and therefore the capital
requirement. This assumption is crucial for the constrained problem
\eqref{shoptrsreg} to be well-posed. See also Remark~\ref{remconv}
below. The assumption of monotonicity is economically meaningful
since an annual loss greater in any state of the world should lead
to a higher capital requirement. Property {\ref{assrho2}} is of
technical nature. It implies that
\[ \alpha_\rho(p)=\sup\left\{ \alpha\in [0,1]\mid
\rho(L(\alpha,p))\le c_0\right\}\] satisfies either
$\rho(L(\alpha_\rho(p),p))=c_0$, or $\alpha_\rho(p)=1$ if
$\rho(L(\alpha,p))\le c_0$ for all $\alpha\in [0,1]$, or
$\alpha_\rho(p)=-\infty$ if $\rho(L(\alpha,p))> c_0$ for all
$\alpha\in [0,1]$.\footnote{Following the usual convention, we
define $\sup\emptyset=-\infty$.} Since the shareholder prefers great
$\alpha$, that is, $\partial_\alpha U_{SH}>0$, we have that $\alpha_\rho(p)$
equals the ${\rm arg\, max}$ in the regulatory constrained
subproblem in \eqref{shoptrsreg} given that it is well-posed, that
is, $\alpha_\rho(p)>-\infty$. Property {\ref{assrhomon}} implies that $\rho(L(\alpha,p))\le \rho(L(\alpha,p'))$ if $p\ge p'$. Hence the domain $\Dcal_\alpha=\{ p\in (-c_0,\infty)\mid \alpha_\rho(p)>-\infty\}$ of $\alpha_\rho$ is either empty or an interval with $\sup \Dcal_\alpha=\infty$. We define the corresponding intervals of feasible utility levels
\begin{align*}
 \Gamma^{reg}_{SH}&=\left\{
U_{SH}(\alpha_\rho(p),p)\mid p\in\Dcal_\alpha\right\}\\
\Gamma^{reg}_{PH}&=\left\{U_{PH}(\alpha_\rho(p),p)\mid p\in\Dcal_\alpha\right\}.
\end{align*}

Here is our existence and uniqueness result for the regulatory
constrained risk shifting problem~\eqref{shoptrsreg}.

\begin{theorem}\label{thmrsreg}
Assume $\Dcal_\alpha\neq\emptyset$. For any reservation utility level $\gamma_{SH}\in \Gamma^{reg}_{SH}$ and $\gamma_{PH}\in \Gamma^{reg}_{PH}$, respectively,  there exists a unique
solution $(\hat{\alpha},\hat{p})$ in $\Pcal$ to the constrained optimization problem \eqref{shoptrsreg} and to the dual problem
\begin{equation}\label{phoptrsreg}
  \begin{aligned}
    \max_{(\alpha,p)\in \Pcal}  & U_{SH}(\alpha,p)\\
    \text{s.t. }& U_{PH}(\alpha,p)\ge \gamma_{PH},\\
    &\alpha\in {\rm arg\, max}_{\alpha'}\, U_{SH}(\alpha',p)\\
    &\qquad\qquad \text{s.t. } \rho(L(\alpha',p))\le c_0
  \end{aligned}
  \end{equation}
respectively. This solution satisfies $U_{SH}(\hat{\alpha},\hat{p})\ge
\gamma_{SH}$ and $U_{PH}(\hat{\alpha},\hat{p})=
\gamma_{PH}$, respectively, and $\rho(L(\hat{\alpha},\hat{p}))= c_0$ if
$\hat{\alpha}<1$.

Moreover, any solution $(\hat{\alpha},\hat{p})$ to \eqref{shoptrsreg} with $\gamma_{SH}=U_{SH}(\hat{\alpha},\hat{p})\in \Gamma^{reg}_{SH}$ is a solution to \eqref{phoptrsreg} with $\gamma_{PH}=U_{PH}(\hat{\alpha},\hat{p})\in  \Gamma^{reg}_{PH}$, and vice versa.

\end{theorem}

\begin{proof}
We argue in the $(v,w)$-coordinates introduced in the proof of
Theorem~\ref{thm2}, and define the corresponding regulator's risk
measurement function
\[ V_R(v,w)=\rho(-w-v R+X)=-w+\rho(-vR+X) \]
for $0< v\le w<\infty$. It follows by inspection
that $V_R(v,w)=\rho(L(\alpha,p))-c_0$, and thus $V_R(v,w)$ is continuous in $v$. Moreover, $\alpha_\rho(p)$
corresponds to
\[ v_\rho(w)=\sup\left\{ v\in [0,w]\mid V_{R}(v,w)\le 0\right\} .\]
In view of Lemmas~\ref{lemUPalpha} and \ref{lemgradV}, we know
that $\partial_v V_{SH}>0$. Hence, any solution $(\hat{v},\hat{w})$
to \eqref{shoptrsreg} or \eqref{phoptrsreg} must be of the form $\hat{v}=v_\rho(\hat{w})$.

We claim that $v_\rho$ is a non-decreasing function on
$(0,\infty)$. Indeed, arguing by contradiction, suppose that
$v_\rho(w')<v_\rho(w)$ for some $w'>w$. Then $v_\rho(w)\in [0,w]$
and there exists some $v\in (v_\rho(w'),  v_\rho(w)]\neq \emptyset$
with $V_R(v,w)\le 0$. By cash-invariance of $\rho$ it follows
$V_R(v,w')=V_R(v,w)+w-w'<V_R(v,w)\le 0$, which again implies $v\le
v_\rho(w')$, a contradiction. Whence $v_\rho$ is non-decreasing.
Moreover, since $\rho$ and thus $V_R$ is convex, the function
$v_\rho:(0,\infty)\to [-\infty,\infty)$ is concave, and thus
continuous on the interior of its domain $\Dcal=\{ w\in (0,\infty)\mid v_\rho(w)>-\infty\}$ (see e.g.\ Rockafellar\
\cite[Theorems 5.3 and 10.1]{roc_70}).

By assumption, $\Dcal$ is a non-empty interval with $\sup\Dcal=\infty$. We now analyze the properties of the policy- and shareholder utility functions along the curve $(v_\rho(w),w)$ for $w\in \Dcal$. First, we claim that $V_{PH}(v_\rho(w),w)$ is strictly quasiconcave and continuous in $w\in \Dcal$. Indeed, let $  w_1<w_2<w_3$ be points in $\Dcal$. Since $v_\rho$ is concave and non-decreasing, there exists some $\lambda \in (0,1)$ such that $v_\rho(w_2) = \lambda v_\rho(w_1)+(1-\lambda) v_\rho(w_3)$ and $w_2\le \lambda w_1+(1-\lambda)w_3$. From this, and since $\partial_w V_{PH}<0$, we derive
\[ V_{PH}\left(v_\rho(w_2),w_2\right)\ge V_{PH}\left(\lambda (v_\rho(w_1),w_1)+(1-\lambda) (v_\rho(w_3),w_3)\right).\]
On the other hand, the policyholder utility function, $V_{PH}\left( a+bw,w\right)$, is strictly concave along straight lines of the form $(a+bw,w)\in\Vcal$, for constant parameters $a$ and $b$.\footnote{The function  $u\left(w_0+c_0-w-(X-w-(a+bw)R))^+\right)$ is concave in $w$, and strictly concave in $w$ for solvency states $X\le w+(a+bw) R$. Taking expectation preserves the strict concavity since $\Pa[\Scal(a+bw,w)]>0$.} Hence
\[ V_{PH}\left(\lambda (v_\rho(w_1),w_1)+(1-\lambda) (v_\rho(w_3),w_3)\right)>\min\left\{V_{PH}(v_\rho(w_1),w_1),V_{PH}(v_\rho(w_3),w_3) \right\}.\]
We thus obtain $V_{PH}\left(v_\rho(w_2),w_2\right)>\min\left\{V_{PH}(v_\rho(w_1),w_1),V_{PH}(v_\rho(w_3),w_3) \right\}$, which proves the quasiconcavity of $V_{PH}(v_\rho(w),w)$. The continuity follows from the continuity of $v_\rho$ on $\Dcal$. On the other hand, since $\partial_v V_{SH}>0$, $\partial_w V_{SH}>0$, and $v_\rho$ is concave and non-decreasing, we conclude that $V_{SH}(v_\rho(w),w)$
is a strictly increasing continuous function in $w\in
\Dcal$.

As shown in Lemma~\ref{lemUPHmax} {\ref{lemUPHmax3}}, the
policyholder's level set $\{V_{PH}\ge \gamma_{PH}\}$ is compact
in $\Vcal$ for any $\gamma_{PH}$. The theorem now follows by the afore proved properties of the policy- and shareholder utility functions along the curve $(v_\rho(w),w)$, $w\in\Dcal$.
\end{proof}

\begin{remark}\label{remconv}
We note that without convexity of $\rho$, the function $v_\rho(w)$
and thus $V_{PH}(v_\rho(w),w)$ and $V_{SH}(v_\rho(w),w)$ may fail to be continuous in $w$.
Therefore the maximum of \eqref{shoptrsreg} or \eqref{phoptrsreg} may not be attained.
\end{remark}

Theorem~\ref{thmrsreg} shows that there exists a unique solution to the risk
shifting problem under the regulatory constraints. However, while the policyholder constraint is binding at this
investment and premium policy, the shareholder's participation constraint may not. The regulatory constraint is binding for inner solutions ($0\le \hat{\alpha}<1$), but not necessarily for $\hat{\alpha}=1$. From the proof of Theorem~\ref{thmrsreg} we obtain the following corollary.

\begin{corollary}\label{corrsreg}
Assume $\Dcal_\alpha\neq\emptyset$. For any reservation utility level $\gamma_{SH}\in \Gamma^{reg}_{SH}$, there exists a unique policy $(\hat{\alpha}',\hat{p}')$ in $\Pcal$ which lies on the intersection of the shareholder level curve, $U_{SH}(\hat{\alpha}',\hat{p}')=\gamma_{SH}$, and the regulatory constraint set in \eqref{shoptrsreg}. That is,
\[ U_{SH}(\hat{\alpha}',\hat{p}')=\max\left\{U_{SH}(\alpha',\hat{p}') \mid \rho(L(\alpha',\hat{p}'))\le c_0\right\}.\]
The maximizer $(\hat{\alpha},\hat{p})$ of \eqref{shoptrsreg} satisfies $U_{PH}(\hat{\alpha},\hat{p})\ge U_{PH}(\hat{\alpha}',\hat{p}')$.
\end{corollary}


We now consider the effect of solvency regulation on the agency cost
by comparing the investment and premium policies and their implied
welfare under the risk shifting problem without and with solvency
regulation. While in view of Theorems~\ref{thm1} and \ref{thmrs}, the shareholder's
reservation utility constraint is binding without solvency
regulation, it may not be binding under solvency regulation (see
Theorem~\ref{thmrsreg}). However, as indicated in Corollary~\ref{corrsreg}, it is sufficient for welfare comparison to compare the policyholder's utility as a function of
$\alpha\in [0,1]$ along the shareholder's respective level curve.

In the proof of Theorem~\ref{thm3}, we have shown that this function
can only assume three possible shapes. It is either strictly
increasing, or strictly decreasing, or attains a global maximum at a
unique critical point $\alpha^\ast\in [0,1]$ and is strictly
increasing to the left and strictly decreasing to the right of
$\alpha^\ast$.

If the policyholder's utility along the shareholder's respective
level curve is strictly increasing in $\alpha$, then the investment
policy under the risk shifting problem without solvency regulation,
$\bar{\alpha}=1$, is Pareto optimal, i.e.
$\alpha^\ast=\bar{\alpha}=1$. Solvency regulation by limiting the
investment policy to $\hat{\alpha}'$ with
\[\hat{\alpha}'< \alpha^\ast=\bar{\alpha}=1\]
reduces the policyholder's utility and thus welfare.

If the policyholder's utility along the shareholder's respective
level curve attains a global maximum at a unique critical point
$\alpha^\ast\in [0,1)$, the risk shifting
problem leads to a welfare loss. If solvency regulation
implies an investment level $\hat{\alpha}'$ which is higher than the
Pareto optimal level, i.e. if
\[ \alpha^\ast\le\hat{\alpha}'<\bar{\alpha}=1,\]
then the policyholder's utility and thus welfare is higher under the
regulatory constraint. This is because the policyholder's utility is
strictly decreasing for all $\alpha\ge\alpha^\ast$. If solvency
regulation is tighter and implies an investment level which is lower
than the Pareto optimal level, i.e. if
\[ \hat{\alpha}'< \alpha^\ast<\bar{\alpha}=1,\]
then the impact of the regulatory constraint on welfare is
ambiguous. This is because the policyholder's utility is strictly
increasing for all $\alpha\le \alpha^\ast$. In particular, for very
tight solvency regulation that restricts the set of investment
policies to very low levels, regulation can even further reduce
welfare relative to the risk shifting problem without regulation.

In view of the dual problem in Theorem~\ref{thmrsreg}, the welfare
effect of solvency regulation can be analogously discussed by
comparing the shareholder's utility as a function of $\alpha\in
[0,1]$ along the policyholder's respective level curve. Since this
function can also only assume the same three shapes (see
Remark~\ref{rem3}), we obtain the qualitatively identical results.


\section{Numerical Example}\label{secnum}

In this section we first calibrate our model to the average portfolio
of an European Economic Area non-life insurer taken from the Quantitative
Impact Study 3 (QIS3) Benchmarking Study~\cite{crof} of the Chief Risk
Officer (CRO) Forum. From these data we derive all exogenous model parameters in our model.
In particular, we determine the initial
capital $c_0$ and the stochastic model for market risk $R$ and insurance risk $X$.
We then use our numerical findings to illustrate our analytical results.

The average stand alone capital requirements for stock market
investment and insurance risk under the Solvency II standard model
\cite{ceiops3,ceiops4} are
\[ {\rm SCR}_{mkt}=2,508\quad \text{and} \quad {\rm SCR}_{ins}=4,332,\]
respectively.\footnote{These figures are derived from the proportion
splits of QIS3 capital charges as shown on pages 39, 41, 43 in the
document \cite{crof}. The capital requirements are thus normalized
such that the undiversified total solvency capital requirement (SCR)
results in $100\times 100$. The risk class ``default'' is negligible
and the market risk types other than ``equity'' have been omitted
for simplicity. The final numbers can be extracted from
\cite{fil_09}, Figure 2 where ${\rm SCR}_{ins}$ is derived from the
SCR for premium risk and catastrophe risk which are assumed to be
uncorrelated.} Under Solvency II, the market investment and
insurance risk are assumed to have a linear correlation coefficient
of $0.25$. Thus, the diversified total solvency capital requirement
equals
\[ {\rm SCR}_{tot} = \sqrt{{\rm SCR}_{mkt}^2+2\times 0.25\times
{\rm SCR}_{mkt}\,{\rm SCR}_{ins}+{\rm SCR}_{ins}^2}=5,522.\]
The Solvency II risk measure $\rho$ is the value-at-risk, ${\rm
VaR}_{99.5\%}$, at the $99.5\%$ confidence level. Thus, the Solvency
II stand alone capital requirement for market risk is determined by
\begin{equation}\label{eqkmkt}
 {\rm SCR}_{mkt}={\rm VaR}_{99.5\%}\bigl[ \text{market loss}=-(c_0+p_0)\,\alpha_0\,R\bigr]
\end{equation}
for some representative premium $p_0$ to be determined below, and
the representative investment policy $\alpha_0=1/7$.\footnote{This
number is derived from annual financial statements of non-life
insurers.} The Solvency II stand alone capital requirement for
insurance risk equals
\begin{equation}\label{eqkins}
 {\rm SCR}_{ins}={\rm VaR}_{99.5\%}\bigl[ \text{insurance loss}=X-p_0\bigr].
\end{equation}
The Solvency II test demands that the available capital, $c_0$, be greater than or equal to the total solvency capital requirement, ${\rm SCR}_{tot}$. We henceforth assume that
\begin{equation}\label{eqktot}
    {\rm SCR}_{tot}=c_0.
\end{equation}
We now specify the stochastic model for market investment risk $R$
and insurance risk $X$. We assume that
\[ R=\e^{Y}-1 \quad\text{and}\quad X=\e^{Z},\]
where $(Y,Z)$ is jointly normally distributed with mean
$(\mu_Y,\mu_Z)$, standard deviations $\sigma_Y$, $\sigma_Z$, and a
linear correlation of $-0.25$.\footnote{This yields an approximate linear
correlation of $0.25$ between $-R$ and $X$.} Furthermore, we assume that
\begin{equation}\label{Ercali}
  \E[R]=0.04,\quad {\rm var}[R]=\bigl(0.16\bigr)^2.
\end{equation}
We use the following premium calculation principle to calibrate the
insurance risk parameters
\begin{equation}\label{eqpcp}
  p_0=\E[X].
\end{equation}
In solving Equations~\eqref{eqkmkt}--\eqref{eqpcp} we determine the parameters that
fully specify our model ($c_0$, $\mu_Y$, $\mu_Z$, $\sigma_Y$, $\sigma_Z$).\footnote{
For numerical reasons we normalize Equations~\eqref{eqkmkt}--\eqref{eqpcp} such that
$c_0+p_0=1$. This does not change the qualitative aspects of our numerical results.}

As for the policyholder's utility function, we assume constant
absolute risk aversion.\footnote{We choose constant absolute risk aversion since it allows for negative terminal wealth. Moreover, the risk preferences are independent of the policyholder's initial wealth $w_0$. Below we choose $w_0$ in a way that it increases numerical precision.} Thus, the policyholder's expected utility is
\[ U_{PH}(\alpha,p)=-\e^{-\beta w_0}\,\E\left[  \e^{-\beta\left(-p-(X-(c_0+p)(1+\alpha R))^+ \right)
}\right],\]
where $\beta$ denotes the coefficient of absolute risk aversion.

Our numerical results are presented in Figures~\ref{figg10}, \ref{figg30}, and \ref{figg130}
in Appendix~\ref{appfig} for different degrees of risk aversion, $\beta=10,\,30,$ and $130$,
respectively. Computation shows that Assumption~\eqref{assrs} {\ref{assrs1}} is satisfied for all three values of $\beta$. Moreover, Lemma~\ref{lemnumass} ensures that also Assumption~\eqref{assrs} {\ref{assrs2}} holds. Thus, in view of Lemmas~\ref{lemUPalpha} and \ref{lemgrad}, the shareholder's utility is increasing in $\alpha$ and $p$.
The plots cover a section of the policy space $\Pcal$ where
$\alpha$ is on the horizontal and $p$ on the vertical axis. The thin solid and thin dashed lines depict the level curves of
the shareholder (LC~Sh) and policyholder (LC~Ph) for utility levels $\gamma_{SH}$ and $\gamma_{PH}$. The policyholder utility is maximal in the south--west corner of the plots. The thick line characterizes the policy strategies $(\alpha,p)$ that
satisfy the first order condition~\eqref{foceq} for Pareto
optimality (FOC). We observe that in accordance with
Theorems~\ref{thm2} and \ref{thm3} there exists at most one Pareto
optimum for every $\alpha$ that is unique for a specific reservation
utility level.

The Pareto optimal policy under perfect competition is labeled with PC.
It lies on the shareholder's reservation utility level that is equal to his
outside option of not selling insurance, $\gamma_{SH}=c_0$. Analogously,
MO labels the Pareto optimal policies in a monopolistic insurance
market which is located on the policyholder's reservation utility level equal to
his outside option of not buying insurance, that is, $\gamma_{PH}=\E\left[u(w_0-X)\right]\in\Gamma_{PH}$ by Lemma~\ref{lemUPHmax}. In a frictionless market, depending on the level of competition the equilibrium solution lies
on the FOC between PC and MO.

The dotted and slash-dotted lines are the boundaries of the regulatory constraints
$\rho\left(L(\alpha,p)\right)\le c_0$ under the value-at-risk
measure ${\rm VaR}_{99.5\%}$ (VaR) and the expected shortfall measure
${\rm ES}_{99\%}$ (ES), respectively. The policies that are acceptable to
the regulator are to the north-west of these boundaries.
We note that the value-at-risk, ${\rm VaR}_{99.5\%}$, does not
satisfy the convexity property in Assumption~\ref{assrho} in
general, see e.g.\ McNeil et al.\ \cite{mfe}. However, in our
example it shows convex behavior for the relevant values of
$(\alpha,p)$. For comparison, we also consider the Swiss Solvency
Test \cite{sst} regulatory risk measure, which is the expected
shortfall\footnote{Also known as conditional or tail
value-at-risk.}, ${\rm ES}_{99\%}$, at the $99\%$ confidence level.
The expected shortfall satisfies all properties of Assumption~\ref{assrho}.

Figure~\ref{figg30-2} shows the regulatory constraints under the expected shortfall
measure ${\rm ES}_q$ at different confidence levels $q=99,\,90,$ and $60$. In
a frictionless market with perfect competition the optimal investment and
premium strategy $(\alpha^\ast,p^\ast)$ will be obtained at point PC.
However, if shareholders cannot credible commit to an investment strategy
the optimal solution $(\bar{\alpha},\bar{p})$ is attained at the risk shifted solution RS with $\bar{\alpha}=1$ as implied by Theorem~\ref{thmrs}. This is harmful to the policyholder since the policyholder's utility decreases
as $(\alpha,p)$ moves away from the PC along the shareholder's reservation
utility curve.
In this case regulation helps. Regulation restricts the set of
feasible premium and investment strategies and might keep the
shareholder from excessive risk taking. According to Theorem~\ref{thmrsreg} there exists a
unique solution $(\hat{\alpha},\hat{p})$ for the regulated risk shifting problem.
For different confidence levels these solutions are labeled R99, R90, and R60,
respectively.\footnote{The numerical values for $\alpha$, $p$ and $U_{PH}$ at
points of interest can be looked up in Table~\ref{tab}.} As can be seen from Figure~\ref{figg30-2}, regulation
improves efficiency under the risk shifting problem for confidence levels
$q\in[99,90,60]$. Among the
regulatory measures presented in Figure~\ref{figg30-2}, ${\rm ES}_{90\%}$ is
optimal. Solvency requirements are too tight for ${\rm ES}_{99\%}$ and too
weak for ${\rm ES}_{60\%}$.

We have chosen constant absolute risk aversion because it eliminates wealth effects from our analysis. Unfortunately, the
coefficient of absolute risk aversion is difficult to interpret since most empirical studies on
estimating risk aversion are based on the assumption of constant relative risk aversion.
Nevertheless, we can observe the behavior of the Pareto optimal
point under perfect competition.
We find that, with increasing degree of risk aversion $\beta$,
the optimal investment in the stock market reduces
while the optimal premium level increases. Moreover, in our
example, the expected shortfall measure implies a more stringent
regulatory requirement than the one implied by the value-at-risk
measure. For lower degrees of risk aversion, $\beta=10$ and
$\beta=30$, the Pareto optimal policies do
not satisfy the regulatory constraints. For higher degrees of risk
aversion, e.g. $\beta=130$, PC meets the regulatory requirements.

\section{Conclusion}\label{secconcl}

In this paper, we provide a formal framework to analyze the conflict
of interest policyholders and shareholders of insurance companies
face. Increasing the risk of the insurer's assets and liabilities
raises shareholder value potentially at the expense of
policyholders. We characterize investment strategies and premium
policies under Pareto optimality, under the risk shifting problem,
and under solvency regulation. Moreover, we analyze the effect of
solvency regulation on the agency cost of the risk shifting problem.

Solvency capital requirements limit the set of possible risk
structures and thereby provide a commitment device for shareholders
imposed by the regulator. There are other possible contractual or
organizational arrangements that serve as commitment devices and
thus reduce the agency cost. Examples include investment and
dividend policy restrictions in corporate charters, issuing
participating policies, or changing the organizational form to a
mutual insurance company. While these contractual arrangements
reduce the agency cost of the risk shifting problem they add other
trade-offs. In this paper, we took the insurance contract as given
and explored how solvency regulation might address the risk shifting
problem without distorting insurance contracts.

\begin{appendix}
\section{Appendix: Lemmas}

\begin{lemma}\label{lemvexcav}
\begin{enumerate}
    \item $U_{SH}(\alpha,p)$ is convex in $\alpha$ and in $p$.
\item $U_{PH}(\alpha,p)$ is concave in $\alpha$ and strictly concave in $p$.
\end{enumerate}
\end{lemma}

\begin{proof}
For fixed $R=r$ and $X=x$, the function $((c_0+p)(1+\alpha R)-X)^+$
is convex in $p$ and in $\alpha$. Moreover, $u\left(w_0-p-(X-(c_0+p)(1+\alpha R))^+\right)$ is
concave in $p$ and in $\alpha$, and strictly concave in $p$ for solvency states $X\le(c_0+p)(1+\alpha R)$. Taking expectation preserves these properties since $\Pa[\Scal(\alpha,p)]>0$.
\end{proof}

\begin{lemma}\label{lemgrad}
The derivatives of $U_{SH}$ and $U_{PH}$ are given by:
\begin{align*}
  \partial_\alpha U_{SH}(\alpha,p)&=(c_0+p) \E\left[R\,1_{\Scal(\alpha,p)}
  \right]\\
\partial_p U_{SH}(\alpha,p)&=  \E\left[(1+\alpha R)\,1_{\Scal(\alpha,p)}
  \right]>0\\
\partial_\alpha U_{PH}(\alpha,p)&=(c_0+p) \E\left[R\,u'(w_0-p-X+(c_0+p)(1+\alpha
R))\,1_{{\Scal(\alpha,p)}^c}\right]\\
\partial_p U_{PH}(\alpha,p)&=-u'(w_0-p)\Pa \left[{\Scal(\alpha,p)}
  \right]\notag\\
  &\quad+\alpha\E\left[R\,u'(w_0-p-X+(c_0+p)(1+\alpha
R))\,1_{{\Scal(\alpha,p)}^c}\right]
\end{align*}
for all $(\alpha,p)\in \Pcal$.
\end{lemma}

\begin{proof}
We can write
\begin{align*}
  U_{SH}(\alpha,p)&=\int_{-\infty}^{\infty}\int_{-\infty}^{(c_0+p)(1+\alpha r)}\left((c_0+p)(1+\alpha
r)-x\right) f(x,r)\,dx\,dr\\
 U_{PH}(\alpha,p)&=\int_{-\infty}^{\infty}\int_{-\infty}^{(c_0+p)(1+\alpha r)} u(w_0-p) f(x,r)\,dx\,dr\\
 &\quad +
\int_{-\infty}^{\infty}\int_{(c_0+p)(1+\alpha r)}^{\infty} u(
w_0-p-x+(c_0+p)(1+\alpha r)) f(x,r)\,dx\,dr
\end{align*}
Hence the assertion follows by straightforward formal
differentiation as justified by Assumption~\ref{ass1}. That
$\partial_p U_{SH}(\alpha,p)>0$ follows from
Assumption~\ref{ass1}\ref{ass1fxr}.
\end{proof}

\begin{lemma}\label{lemUPHmax}
\begin{enumerate}

  \item\label{lemUPHmax1} $U_{SH}(\alpha,-c_0)\equiv 0$ and $\lim_{p\to\infty} \inf_{\alpha\in[0,1]} U_{SH}(\alpha,p)=\infty$. Hence the range is $U_{SH}(\Pcal)=(0,\infty)$.

\item\label{lemUPHmax2} For every level $\gamma_{SH}\in (0,\infty)$, the level curve $\{U_{SH}=\gamma_{SH}\}$ is a $C^1$-line in $\Pcal$ connecting the $\{\alpha=0\}$- and $\{\alpha=1\}$-axes.

\item\label{lemUPHmax3} $U_{PH}(\alpha,-c_0)\equiv \E\left[u(w_0+c_0-X)\right]$ and $\lim_{p\to\infty} \sup_{\alpha\in[0,1]}U_{PH}(\alpha,p)=-\infty$.

\item\label{lemUPHmax4} $U_{PH}$ attains its global supremum at $(1, p^{crit}_{PH})$ . That is,
    \begin{equation}\label{UPHmaxi}
   U_{PH}(1, p^{crit}_{PH})=\sup_{(\alpha,p)\in  {\Pcal}} U_{PH}(\alpha,p) .
    \end{equation}
Moreover, $p^{crit}_{PH}>-c_0$ if and only if $\E\left[R\,u'(w_0+c_0-X)\right]>0$.\footnote{A sufficient condition for this to hold is $\E\left[ R\mid X=x\right]>0$ for all $ {\rm ess \,inf}\, X<x<{\rm ess \,sup}\, X$.} In this case, $(1, p^{crit}_{PH})$ is the unique maximizer for $U_{PH}$ in $\Pcal$. Hence the range is $U_{PH}(\Pcal)=\Gamma_{PH}$.

\item\label{lemUPHmax5}  For every level $\gamma_{PH}\in \Gamma_{PH}$, the level curve $\{U_{PH}=\gamma_{PH}\}$ is a $C^1$-line in $\Pcal$, connecting the $\{\alpha=0\}$- and $\{\alpha=1\}$-axes if $\gamma_{PH}\le \E\left[u(w_0+c_0-X)\right]$.
\end{enumerate}
\end{lemma}

\begin{proof}
Most of the stated properties follow straightforward from the definition of $U_{SH}$ and $U_{PH}$, and Lemmas~\ref{lemvexcav} and \ref{lemgrad}. We only give details for some of the statements.

Jensen's inequality implies $U_{SH}(\alpha,p)\ge (c_0+p)(1+\alpha\E[R])\ge (c_0+p)\min\{1,1+\E[R]\}$. Since $\E[R]>-1$ we conclude that $\lim_{p\to\infty} \inf_{\alpha\in[0,1]} U_{SH}(\alpha,p)=\infty$, which proves {\ref{lemUPHmax1}}.

Next, note that $U_{PH}(\alpha,p)\le \E[u(w_0-p)]$, hence $\lim_{p\to\infty}\sup_{\alpha \in [0,1]} U_{PH}(\alpha,p)=-\infty$, and $U_{PH}$
is uniformly bounded on $\Pcal$ in particular, whence {\ref{lemUPHmax3}}.

In view of Lemma~\ref{lemgrad},
$U_{PH}$ cannot have a critical point in the interior of $\Pcal$,
and $U_{PH}(0,p)$ is strictly decreasing in $p$. On the other hand,
$U_{PH}(\alpha,-c_0)\equiv \E[u(w_0+c_0-X)]$ is constant in
$\alpha\in [0,1]$. We conclude that $U_{PH}(\alpha,p)$ attains its
maximum in $(1, p^{crit}_{PH})$, which proves \eqref{UPHmaxi}. Next we note that $\Pa[\Scal(\alpha,-c_0)]=\Pa[ X\le 0]=0$. Hence Lemma~\ref{lemgrad} implies that $\partial_p U_{PH}(\alpha,-c_0)=\alpha\E\left[R\,u'(w_0+c_0-X)\right]$, and thus $p^{crit}_{PH}>-c_0$ if and only if $\E\left[R\,u'(w_0+c_0-X)\right]>0$, whence {\ref{lemUPHmax4}}.

The smoothness properties {\ref{lemUPHmax2}} and {\ref{lemUPHmax5}} of the level curves follows by the implicit function theorem and the fact that $\nabla U_{SH}$ and $\nabla U_{PH}$ are nowhere zero in $\Pcal$. Finally, $U_{PH}(0,p)$ and $U_{PH}(1,p)$ are strictly decreasing in $p\in (-c_0,\infty)$ and $p\in (p^{crit}_{PH},\infty)$, respectively, with $U_{PH}(0,-c_0)=\E\left[u(w_0+c_0-X)\right]\le \gamma^{crit}_{PH}=U_{PH}(1,p^{crit}_{PH})$. Whence level curve $\{U_{PH}=\gamma_{PH}\}$ connects the $\{\alpha=0\}$- and $\{\alpha=1\}$-axes if $\gamma_{PH}\le \E\left[u(w_0+c_0-X)\right]$. The statement on shareholder level curves follows similarly.
\end{proof}

\begin{lemma}\label{lemgradV}
The derivatives of $V_{SH}$ and $V_{PH}$ defined in
\eqref{vwnotation} are given by:
\begin{align*}
  \partial_v V_{SH}(v,w)&=  \E\left[R\,1_{\Scal(v,w)}
  \right]\\
\partial_w V_{SH}(v,w)&=  \Pa \left[{\Scal(v,w)}
  \right]>0\\
\partial_v V_{PH}(v,w)&=  \E\left[R\,u'(w_0+c_0-X+vR)\,1_{{\Scal(v,w)}^c}\right]\\
\partial_w V_{PH}(v,w)&=-u'(w_0+c_0-w)\Pa \left[{\Scal(v,w)} \right]<0
\end{align*}
for all $0\le v\le w$, $w\in [c_0,\infty)$.
\end{lemma}

\begin{proof}
Follows from Lemma~\ref{lemgrad} and Assumption~\ref{ass1}.
\end{proof}

\begin{lemma}\label{lemnumass}
Let
\[ R=\e^{Y}-1 \quad\text{and}\quad X=\e^{Z},\]
where $(Y,Z)$ is jointly normal distributed with mean
$(\mu_Y,\mu_Z)$, standard deviations $\sigma_Y$, $\sigma_Z$, and
linear correlation coefficient $\rho_{(Y,Z)}$. If $\E[R]>0$ and
$\rho_{(Y,Z)}<0$, then
\[\E\left[ R\mid X\le c\right]>0\quad\text{for all
  $c\in (0,\infty)$.}\]
\end{lemma}

\begin{proof}
The conditional distribution of $Y$ given $Z$ is normally
distributed with mean $\mu_{Y\mid
Z=z}=\mu_Y+\rho_{(Y,Z)}\frac{\sigma_Y}{\sigma_Z}( z-\mu_Z) $ and
variance $\sigma^2_{Y\mid Z=z}=\sigma_Y^2( 1-\rho_{(Y,Z)}^2)$, see
e.g.\ McNeil et al.\ \cite[p. 68]{mfe}. Hence
\[\E[R\mid X=x]=\int_{-\infty}^{\infty}\left( \e^y - 1\right)\phi_{\mu_{Y\mid Z=\ln (x)}, \sigma^2_{Y\mid Z=\ln(x)}} (y) dy\]
where $\phi_{\mu,\sigma}$ denotes the normal density function with
mean $\mu$ and variance $\sigma^2$. Since $(e^y-1)$ is strictly
increasing in $y$, $\partial_x\mu_{Y\mid Z=\ln (x)}\le 0$, and
$\partial_x\sigma^2_{Y\mid Z=\ln (x)}= 0$, we infer that
$\partial_x\E[R\mid X=x]\le 0$. Suppose, by contradiction, there
exists some $c \in (0,\infty)$ with $\E[R \mid X \le c] \le 0$. Then
$\partial_x\E[R\mid X=x]\le 0$ implies that $\E[R\mid X=c]\le 0$ and
thus $\E[R\mid X=x]\le 0$ for all $x \ge c$. Hence $\E[R]=\E[R \mid
X \le c]\Pa[X \le c] + \E[R \mid X > c]\Pa[X>c] \le 0$ contradicts
$\E[R]>0$.
\end{proof}

\end{appendix}

\pagebreak{}
\section{Appendix: Figures and Table}\label{appfig}
\suppressfloats[] \vfill
\begin{figure}[h]
\begin{centering}
\includegraphics[bb=50bp 200bp 562bp
592bp,clip,width=1\textwidth]{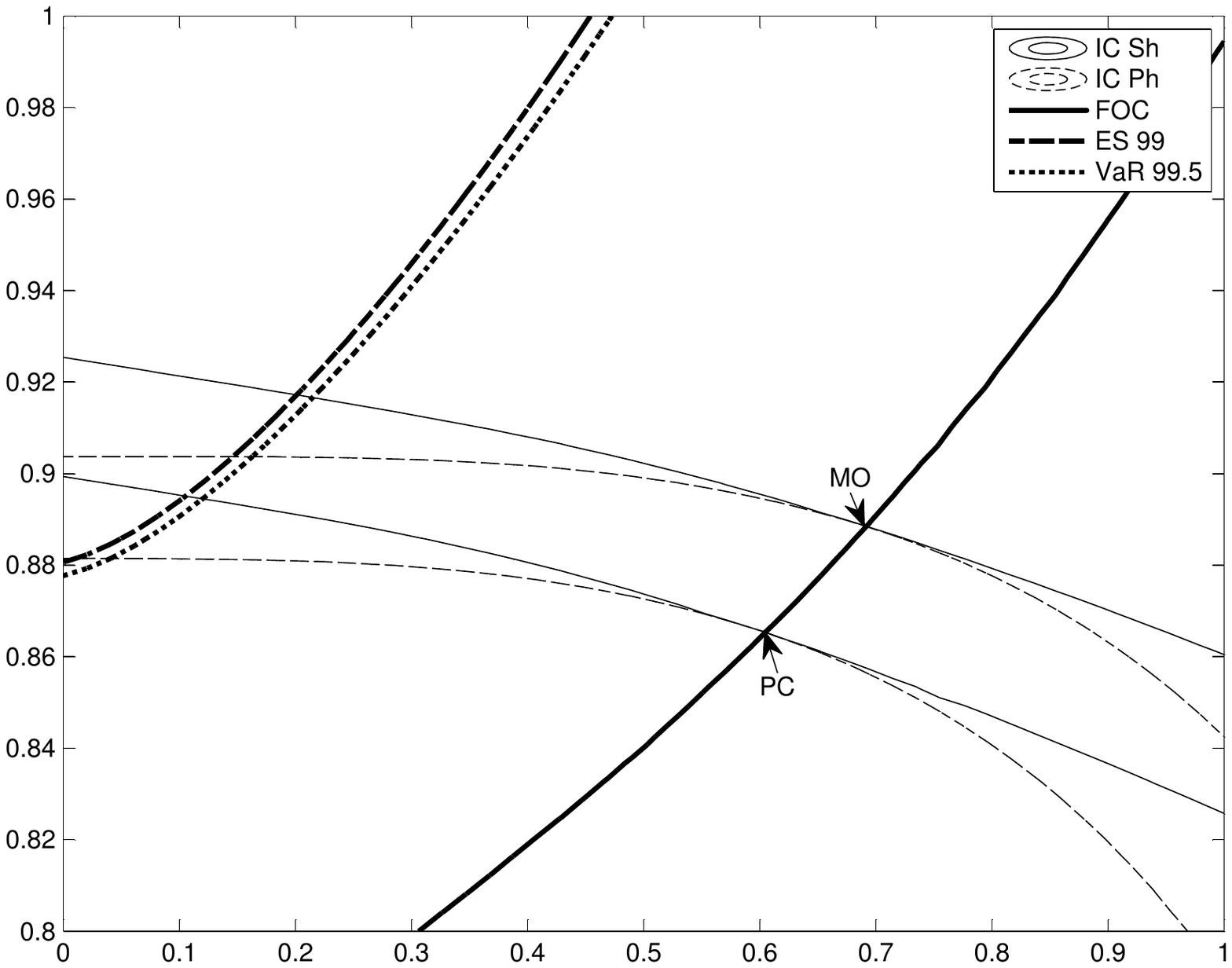}\caption{This Figure
shows the numerical results for a coefficient of absolute risk
aversion $\beta=10$ in the $\left(\alpha,p\right)$-space. The thin
solid lines represent the shareholder's level curve (LC~Sh).
The thin dashed lines represent the policyholder's level
curve (LC~Ph). The contract curve (FOC) is depicted by the thick
solid line. The points PC and MO denote the Pareto optimal policies
in a perfectly competitive and in a monopolistic market,
respectively. The thick dotted and slash dotted lines determine the
set of feasible policies under the ES$_{99\%}$-measure (ES) and the
VaR$_{99.5\%}$-measure (VaR).}\label{figg10}
\end{centering}
\end{figure}
\vfill

\begin{figure}[H]
\begin{centering}
\includegraphics[bb=50bp 200bp 562bp
592bp,clip,width=1\textwidth]{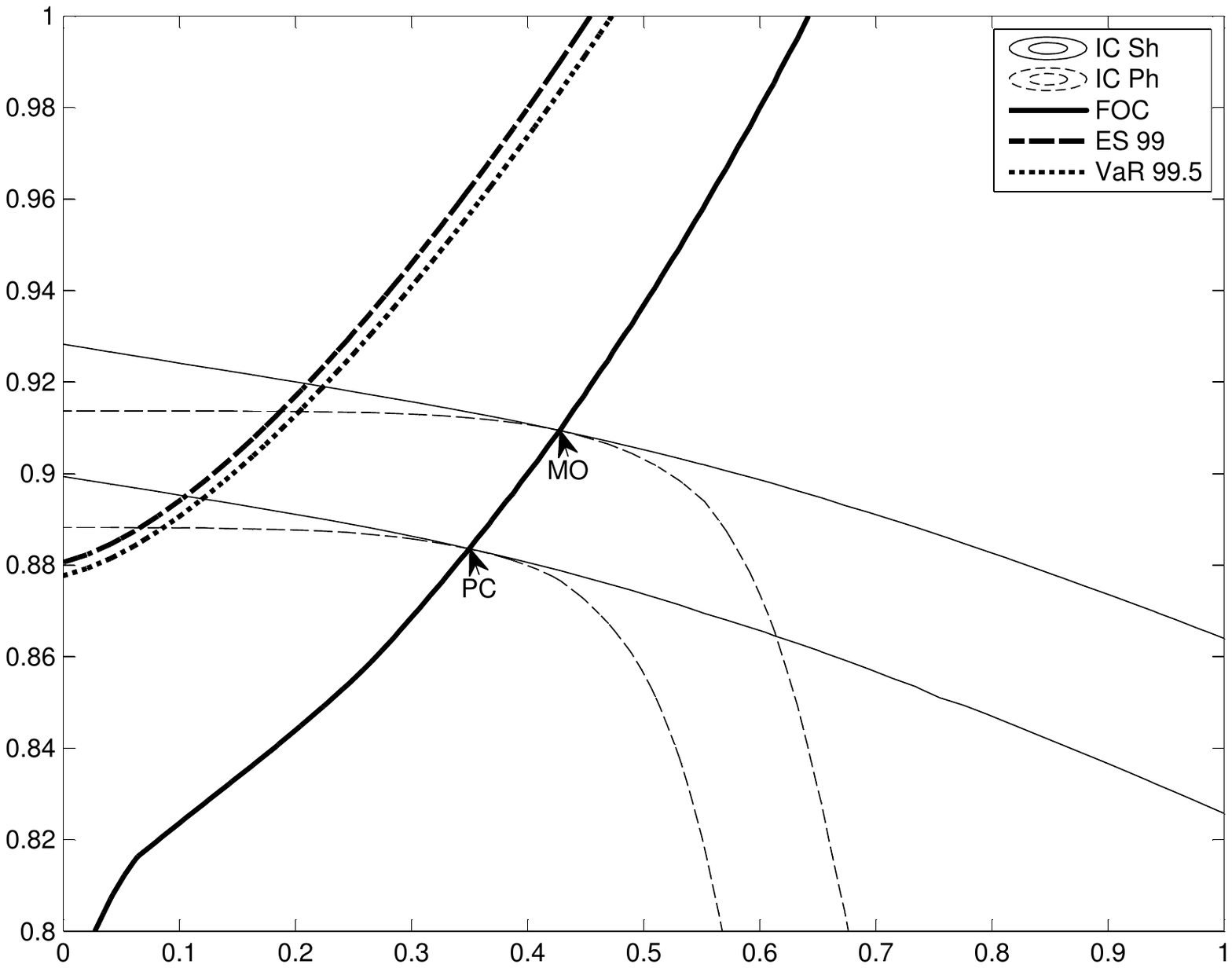}\caption{This Figure
shows the numerical results for a coefficient of absolute risk
aversion $\beta=30$ in the $\left(\alpha,p\right)$-space. The thin
solid lines represent the shareholder's level curve (LC~Sh).
The thin dashed lines represent the policyholder's level
curve (LC~Ph). The contract curve (FOC) is depicted by the thick
solid line. The points PC and MO denote the Pareto optimal policies
in a perfectly competitive and in a monopolistic market,
respectively. The thick dotted and slash dotted lines determine the
set of feasible policies under the ES$_{99\%}$-measure (ES) and the
VaR$_{99.5\%}$-measure (VaR).}\label{figg30}
\end{centering}
\end{figure}

\begin{figure}[H]
\begin{centering}
\includegraphics[bb=50bp 200bp 562bp
592bp,clip,width=1\textwidth]{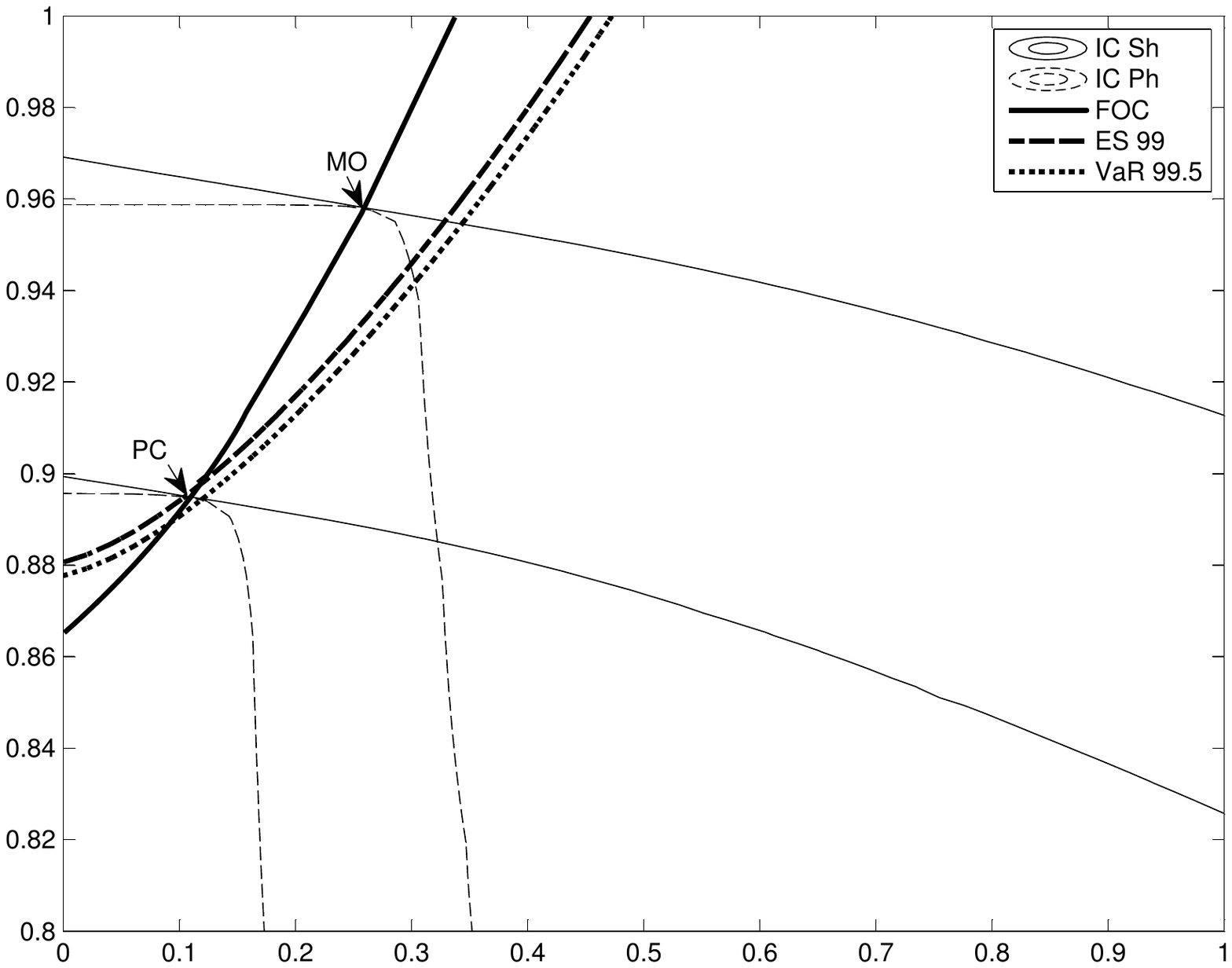}\caption{This Figure
shows the numerical results for a coefficient of absolute risk
aversion $\beta=130$ in the $\left(\alpha,p\right)$-space. The thin
solid lines represent the shareholder's level curve (LC~Sh).
The thin dashed lines represent the policyholder's level
curve (LC~Ph). The contract curve (FOC) is depicted by the thick
solid line. The points PC and MO denote the Pareto optimal policies
in a perfectly competitive and in a monopolistic market,
respectively. The thick dotted and slash dotted lines determine the
set of feasible policies under the ES$_{99\%}$-measure (ES) and the
VaR$_{99.5\%}$-measure (VaR).}\label{figg130}
\end{centering}
\end{figure}

\clearpage

\begin{figure}[H]
\begin{centering}
\includegraphics[bb=50bp 200bp 562bp
592bp,clip,width=1\textwidth]{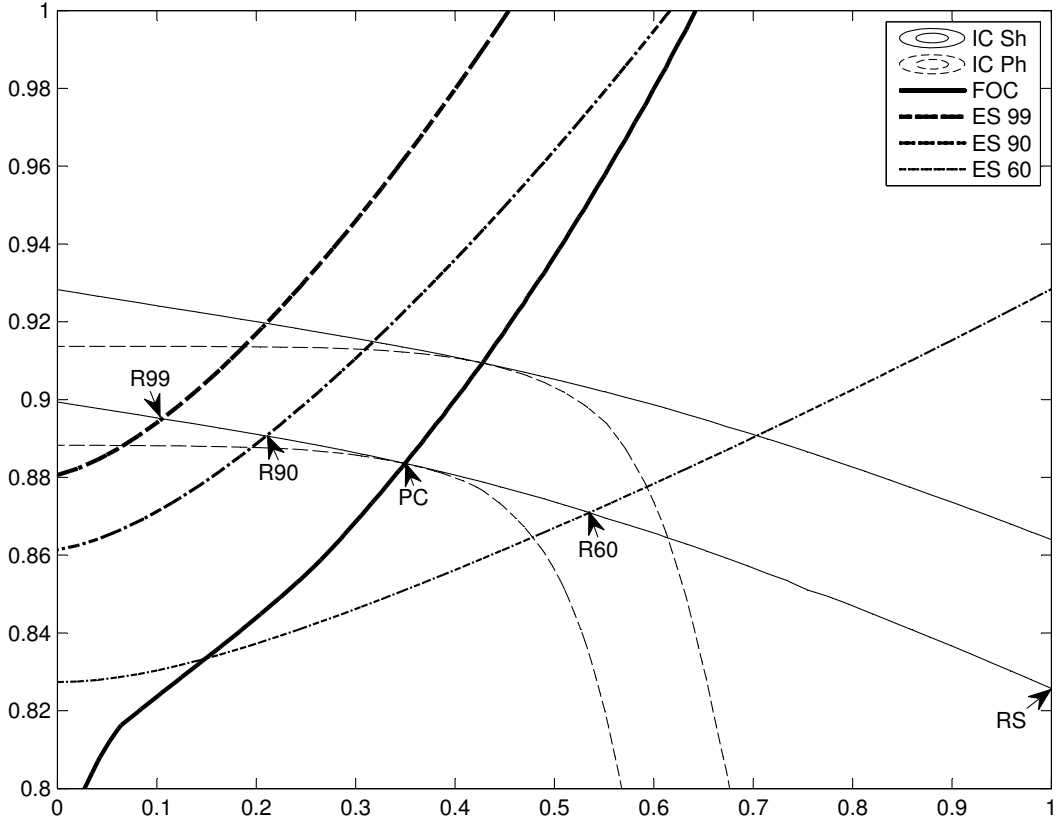}\caption{This Figure
shows the regulatory constraints for a coefficient of absolute risk
aversion $\beta=30$ under the expected shortfall measure ${\rm
ES}_q$ at different confidence levels $q=99,\,90,$ and $60$ for the
case of perfect competition. That is, on the shareholder level curve $\gamma_{SH}=c_0$. In
this setting the Pareto optimal policy $(\alpha^\ast,p^\ast)$ is
labeled with PC, the risk shifted solution is denoted by RS, and the
solutions to the regulated risk shifting problem for different
confidence levels are marked with R99, R90, and R60,
respectively.}\label{figg30-2}
\end{centering}
\end{figure}

\begin{table}[H]
\center
\begin{tabular}{|l|c|c|c|}
  \hline
    & $\alpha$ & $p$ & $U_{PH}(\alpha,p)$ \\ \hline
  PC $=(\alpha^\ast,p^\ast)$ & 0.347 & 0.883 & -1.00E-06 \\ \hline
  R99 $=(\hat{\alpha}_{{\rm ES}_{99\%}},\hat{p}_{{\rm ES}_{99\%}})$ & 0.102 & 0.893 & -1.14E-06 \\ \hline
  R90 $=(\hat{\alpha}_{{\rm ES}_{90\%}},\hat{p}_{{\rm ES}_{90\%}})$ & 0.204 & 0.890 & -1.02E-06 \\ \hline
  R60 $=(\hat{\alpha}_{{\rm ES}_{60\%}},\hat{p}_{{\rm ES}_{60\%}})$ & 0.531 & 0.869 & -1.31E-06 \\ \hline
  RS $=(\bar{\alpha},\bar{p})$ & 1 & 0.824 & -4.21E-05 \\
  \hline
\end{tabular}
\caption{This Table shows the values of the policy $(\alpha,p)$ and the
implied utility of the policyholder for points of interest marked in
Figure~\ref{figg30-2}.}\label{tab}
\end{table}

\end{document}